\documentclass{article}
\usepackage{amsfonts}
\usepackage{amsmath}
\usepackage{amsthm}
\usepackage[margin=1.0in]{geometry}
\usepackage{graphicx}
\usepackage[update,prepend]{epstopdf}
\usepackage{color}

\newtheorem{mypro}{Proposition}
\newtheorem{mytho}{Theorem}
\newtheorem{mylem}{Lemma}

\def\qed{}

\def\bR{{\mathbb{R}}}
\def\bE{{\mathbb{E}}}

\def\eps{\varepsilon}
\def\btheta{\boldsymbol\theta}

\def\calF{{\cal F}}
\def\calG{{\cal G}}
\def\barH{{\overline{H}}}

\def\lhat{\hat{e}}
\def\Lhat{\hat{E}}

\title{Realized volatility and parametric estimation of Heston SDEs}

\date{\today}

\author{
Robert Azencott
\thanks{Department of Mathematics,
University of Houston,
Houston, TX 77204-3008, (\texttt{razencot@math.uh.edu}).}
\and
Peng Ren
\thanks{Department of Mathematics,
University of Houston,
Houston, TX 77204-3008, (\texttt{pren@math.uh.edu}).}
\and
Ilya Timofeyev
\thanks{Department of Mathematics,
University of Houston,
Houston, TX 77204-3008, (\texttt{ilya@math.uh.edu}).}
}

\begin{document}
\maketitle
\begin{abstract}
We present a detailed analysis of \emph{observable} moments based parameter estimators for the Heston SDEs jointly driving the rate of returns $R_t$ and the squared volatilities $V_t$. Since volatilities are not directly observable, our parameter estimators are constructed from empirical moments of realized volatilities $Y_t$, which are of course observable. Realized volatilities are computed over sliding windows of size $\eps$, partitioned into $J(\eps)$ intervals. We establish criteria for the joint selection of  $J(\eps)$  and of the sub-sampling frequency of return rates data. 

We obtain explicit bounds  for the $L^q$ speed of convergence of realized volatilities to true volatilities as $\eps \to 0$. In turn, these bounds provide also $L^q$ speeds of convergence of  our observable estimators for the parameters of the Heston volatility SDE.

Our theoretical analysis is  supplemented by extensive numerical simulations of joint Heston SDEs to investigate the actual performances of our moments based parameter estimators. Our results provide practical guidelines for adequately fitting Heston SDEs parameters to observed stock prices series.

\end{abstract}

\smallskip


\section{Introduction}

Parametric estimation of stochastic differential equations (SDEs) has been an active research area for several decades. The majority of published results focus on \emph{Direct Observability} situations, where the observable data $X_t$ are assumed to be generated by the SDEs themselves. But in many practical situations, the SDEs driving an \emph{unobservable} process $X_t$ are parametrized by a vector $\btheta$ which needs to be estimated from observable data $Y_t^{\eps}$ which are only known to converge to $X_t$ as $\eps \to 0$. We refer to these situations as \emph{Indirect Observability} contexts.
A crucial point is then to assess estimation errors due to the use of approximate data (see, for example, \cite{crom4,past07,papast09,Zhang1,Zhang2,Bandi1,Ilze1}).
In our papers \cite{Azen1,Azen2,Azen3,azreti15}, we analyzed asymptotic consistency of parameter estimation under indirect observability in multiple contexts. In particular, in \cite{azreti15} we proved the asymptotic accuracy of parameter estimators based on empirical moments of indirect approximate observations, for a wide class of unobservable stationary non-Gaussian processes $X_t$ with \lq\lq{}fast\rq\rq{} mixing properties. 
Here we extend and deepen results from \cite{azreti15} to parameter estimation for the well known Heston SDEs \cite{Heston1} driving jointly the rate of returns $R_t$ of an arbitrary asset and its squared volatility $V_t$. Since the volatilities are not directly observable, classical observable approximations of $V_t$ are provided by \emph{realized volatilities} $Y_t^{\eps}$ computed on averaging time windows $(t - \eps, t)$. Such volatility approximations have been studied for instance in \cite{hoffmann2002,Gloter2007,Nielsen1,Christen1,Christen2,sahkim07,Ruiz94,dieb02}.

In this paper focused on feasible parameter estimation for the Heston volatility SDEs, we construct observable parameter estimators from the first and second order empirical moments of the realized volatility process $Y_t^{\eps}$, and analyze their $L^2$-consistency as $\eps \to 0$. In order to do this, we extend the
$L^2$ convergence of the realized volatility to $L^q$ (where $q$ is odd).
Maximum Likelihood estimation and estimates on $L^q$ norms for square-root diffusions 
were also addressed for example in \cite{alke13,bemidi08}.
While the Maximum Likelihood Estimators (MLEs) have been used in many context to estimate parameters of stochastic differential equations, including the Heston model (e.g. \cite{Azencott2015,gloter2000}), MLEs can be quite sensitive to  model errors. We expect moment estimators of low order  to be much more robust with respect to small perturbations of the underlying model fitted to observable and possibly noisy data.
The empirical moments of $Y_t^{\eps}$ rely on explicit sub-sampling schemes which specify key computational parameters (e.g. number of points in the window $(t - \eps, t)$, observational time-step, and total number of observations) as functions of the window size $\eps$. In particular, optimal sub-sampling schemes involve explicit expressions for selecting the observational time-step $\Delta(\eps)$ and
number of observations $N(\eps)$ for the realized volatility. We demonstrate that the
optimal speed of convergence for moment estimators computed under indirect observability  is $O(\eps^{1/2})$.

When realized volatilities are computed over  sliding windows of small duration $\eps$, our target is to determine  nearly optimal stockprice sub-sampling rates enabling good control of estimation errors for the parameters of the Heston SDE driving the (unobservable) volatilities.
In contrast with other results on parametric estimation of the volatility Heston SDE under indirect observability (e.g. \cite{Gloter2007}), our results address estimation of both drift and diffusion parameters in the 
volatility SDE. Parametric estimation of the diffusion term for generic SDEs is a delicate task; several methods have been introduced 
\cite{hoffmann2002,bz2002,gloter2000,comte2009}, but comparing performance and robustness for  various  estimators still remains an active research area.

Application of the general theory developed in \cite{azreti15} requires a substantial analytical investigation of the Heston model. For realized volatilities $Y_t^{\eps}$ computed on sliding windows of length $\eps \to 0$, we give concrete estimates for the $L^q$ convergence speed of $Y_t^{\eps}$ to true volatility $V_t$ and we derive explicit nearly optimal sub-sampling schemes of $Y_t^{\eps}$ for consistent estimation of empirical moments. We compute theoretical convergence speeds for our observable estimators of the Heston SDE parameters as $\eps \to 0$ and we compare them to numerically evaluated convergence speeds. 
To this end, we perform numerical investigation of the Heston model, through extensive simulations with $\eps \to 0$. Our simulations refine and confirm our theoretical convergence rates for the realized volatilities as well as for our estimators of the Heston parameters. We thus validate asymptotically optimal ranges for the number of data points used to compute each realized volatility. Our numerical results indicate that for small but realistic values of $\eps$, nearly optimal convergence rates of our observable parameter estimators can still be achieved with data subsampling less frequent than the theoretically prescribed rates.

We introduce the Heston model and address the
$L^q$-H\"{o}lder continuity for squared volatilities in section \ref{sec3}.
We introduce the realized volatility and the notion of \emph{indirect observability} in section \ref{RealVol}. 
We prove the $L^q$ convergence of realized volatilities in section 
\ref{sec:Lqconverge}.
Some analytic properties of the Heston model are discussed in sections \ref{sec5}, \ref{sec:moments}, and moment-based parameter estimators for the volatility process are introduced in section \ref{sec7}. Theorem \ref{consistency} in section \ref{sec7} is one of our main analytical results, since it computes $L^2$ convergence rates for subsampled empirical moments of realized volatilities and hence yields convergence rates for our observable parameter estimators based on realized volatilities. 
In section \ref{sec8} we discuss pragmatic discretization rates for the averaging windows $(t - \eps, t)$. 
In sections \ref{speedRealVol} and \ref{paramSpeed} we perform an extensive numerical investigation of the Heston model, including numerical computation of the $L^2$ and $L^4$ speeds of convergence for the realized volatility process, speed of convergence of parameter estimators, and empirical covariance estimators.
Conclusions are presented in section \ref{conc}.


\section{The Heston stochastic volatility models}\label{heston}
\label{sec3}
\subsection{Generic stochastic volatility models}
In the well known Barnsdorff-Nielsen paper \cite{Nielsen1}, generic stochastic volatility models consider asset price processes $A_t$ such that the \emph{rate of return} process $dR_t = dA_t / A_t$ is driven by an  SDE of the form
\begin{equation}\label{dRt}
dR_t = \mu dt + \sqrt{V_t} dZ_t,
\end{equation}
where $\mu$ is a constant, $Z_t$ is a standard one-dimensional Brownian motion, and the square integrable continuous process $V_t > 0$ is called  the  \emph{spot variance} or \emph{squared volatility} of the return rate.
In this paper we will focus on the classical \emph{Heston joint SDEs}  which are widely used  examples of stochastic volatility models.

\subsection{The Heston joint SDEs}
Recall that in the Heston model \cite{Heston1} for the stochastic dynamics of asset price $A_t$ and squared volatility $V_t$, two coupled SDEs jointly drive $V_t$ and the rate of return $dR_t = dA_t /A_t$, namely 
\begin{eqnarray}
dR_t &=& \mu dt+\sqrt{V_t}dZ_t,
\label{HesEq1} \\
dV_t &=& \kappa(\theta-V_t)dt + \gamma\sqrt{V_t}dB_t.
\label{HesEq2}
\end{eqnarray}
Here $Z_t$ and $B_t$ are standard one dimensional Brownian motions with constant instantaneous correlation $\bE(dZ_t dB_t) = \rho dt$ where $-1 < \rho < 1$.

The autonomous volatility SDE \eqref{HesEq2} is parametrized by 3 parameters, the \lq\lq{}long run mean\rq\rq{} $\theta > 0 $ of $V_t$, the \lq\lq{}reversion rate\rq\rq{} $\kappa >0 $, and $\gamma > 0 $. 
To ensure that $V_t$ remains almost surely positive for all $t$ provided $V_0$ is almost surely positive, the parameter vector $\btheta = \left[ \kappa, \theta, \gamma \right]$ must verify the classical Feller condition \cite{Feller1}
\begin{equation}
\label{Feller}
\frac{\kappa \theta}{\gamma^2} > \frac{1}{2}.
\end{equation}
\emph{In this paper we assume that parameters in the equation for volatility satisfy the Feller condition above.}
The first Heston SDE \eqref{HesEq1} is parametrized by the constant \lq\lq{}mean return rate\rq\rq{} $\mu >0$ of the asset price, and by the correlation coefficient $\rho$.
To be specific, we define $Z_t = \rho B_t + (1 - \rho^2)^{1/2}\beta_t$, where $B_t$ and $\beta_t$ are two \emph{independent} standard Brownian motions. Let $\calF_t$ be the increasing filtration generated by $B_t$. Let $\calG_t$ be the increasing filtration generated by the pair $(Z_t, B_t)$, or equivalently by $(\beta_t, B_t)$.
Then the joint SDEs \eqref{HesEq1}, \eqref{HesEq2} have a unique solution  $V_t, R_t$  starting at any fixed $V_0>0$ and $R_0$. Moreover, $V_t$ is \emph{nonanticipating} with respect to both $\calF_t$, $\calG_t$,  and $R_t$ is \emph{non-anticipating} with respect to $\calG_t$. Both $V_t$ and $R_t$ are continuous in $t$. 
We will use the shortcut notation
$$
\bE[ U | V_0 = y] \equiv \bE_y [U] \;\; \text{for any random variable $U$}.
$$

\subsection{$L^q$-H\"{o}lder continuity for squared volatilities}
\label{sec:Holder}
We now prove that the solutions $R_t,V_t$ of Heston SDEs  are H\"{o}lder continuous in $L^q$, with H\"{o}lder exponent $1/2$.
\begin{mypro}\label{Lq Holder}
Let $R_t, V_t$ be the returns rate and squared volatility jointly driven by the Heston SDEs \eqref{HesEq1}, \eqref{HesEq2}, starting at any fixed  $V_0 = y >0$ and $R_0 = r$. Fix  any time $T > 0$. Then $(T, y, r)$ and the Heston SDEs parameters determine for each $q \geq 2$ a constant $C = C(q)$ such that, for all $0 \leq s \leq t \leq T$,
\begin{eqnarray}
\| V_t \|_q \leq C \;\;\; \text{and} \;\;\; \| V_t - V_s \|_q \leq C \sqrt{t-s},
\label{HestonHolder1} \\
\| R_t \|_q \leq C \;\;\; \text{and} \;\;\; \| R_t - R_s \|_q \leq C \sqrt{t-s}.
\label{HestonHolder2}
\end{eqnarray}
Moreover, equations \eqref{HestonHolder1} and \eqref{HestonHolder2} still hold when $V_t$ is the unique stationary process driven by the Heston volatility SDE, and $R_0$ is fixed.
\end{mypro}

\begin{proof}

Consider an increasing filtration $\mathcal{W}_t$ and let $W_t$ be a progressively measurable standard Brownian motion with respect to $\mathcal{W}_t$. Let $\mathcal{X}$ be the set of all continuous processes $\{X_t\}$, which are progressively measurable with respect to $\mathcal{W}_t$ and non anticipative with respect to $W_t$.
A classical martingale inequality (see equation (3.1) in \cite{DavisBurgess}) states that for each positive  $q$, there is a universal constant  $H_q$ such that for all $(X_t) \in \mathcal{X}$ and all $0 \leq s \leq t$,
\begin{equation} \label{mart}
\bE \left[  \left(\int_s^t  X_u dW_u \right) ^q \right] \leq  H_q \bE \left[ \left(\int_s^t  X_u^2 du \right)^{q/2} \right], 
\end{equation}
which, denoting $h_q = H_q^{1/q}$, is clearly equivalent to
\begin{equation} \label{mart1}
\left| \left| \;  \int_s^t  X_u \; dW_u \; \right| \right|_q  \; \leq 
h_q \left( \; \left|\left| \; \int_s^t  X_u^2 \; du \; \right| \right|_{q/2} \; \right) ^{1/2}.
\end{equation}
Fix $q \geq 2$,  $T >0$. Let $R_t, V_t$ be  the solutions of the Heston SDEs starting at any fixed $V_0 >0$ and $R_0$. Denote $\calF_t$ and $\calG_t$ the filtrations respectively  generated by $B_t$ and by the pair $(Z_t,B_t)$. Then the solution $V_t$ of the volatility SDE is $\calF_t$ measurable and non-anticipating with respect to $B_t$. The solution $R_t$ of the returns rate SDE is $\calG_t$ measurable and non-anticipating with respect to $Z_t$. 

Next, 
apply \eqref{mart1} first with $\mathcal{W}_t = \calF_t$, the $W_u=B_u$,  $X_u =  \sqrt{V_u}$ , and then a second time with $\mathcal{W}_t = \calG_t$, $W_u=Z_u$, $X_u =  \sqrt{V_u}$. This yields, for all $0 \leq s \leq t $, all $q \geq 2$, 
\begin{equation} \label{mart2}
\left|\left| \int_s^t  \sqrt{V_u} dB_u \right|\right|_q  \leq 
h_q \left( \left|\left| \; \int_s^t  V_u du  \; \right|\right|_{q/2} \right)^{1/2} \leq 
h_q \left( \int_s^t  \| V_u \|_{q/2} du \right)^{1/2}, 
\end{equation}
\begin{equation} \label{mart3}
\left|\left| \int_s^t  \sqrt{V_u} dZ_u \right|\right|_q  \leq 
h_q \left( \left|\left| \; \int_s^t  V_u du  \; \right|\right|_{q/2} \right)^{1/2} \leq 
h_q \left( \int_s^t  \| V_u \|_{q/2} du \right)^{1/2}.
\end{equation}
By Proposition \ref{cond.moments} which is proved further on, there is a constant $c_q$ such that $ \|V_t \|_q \leq c_q$  for all $t$. Hence \eqref{mart2} and \eqref{mart3} imply, with $k_q= h_q \sqrt{ c_{q/2}}$,
\begin{equation} \label{mart6}
\left|\left| \int_s^t  \sqrt{V_u} dB_u \right|\right|_q   \leq k_q  (t-s) ^{1/2},
\end{equation}
\begin{equation} \label{mart7}
\left|\left| \int_s^t  \sqrt{V_u} dZ_u \right|\right|_q   \leq k_q  (t-s) ^{1/2}.
\end{equation}
Integrating the Heston SDEs  \eqref{HesEq1}, \eqref{HesEq2} we obtain
\begin{equation} \label{Vts}
V_t - V_s = (t-s) \kappa \theta  - \kappa \int_s^t V_u du  + \gamma \int_s^t  \sqrt{V_u} dB_u ,
\end{equation}
\begin{equation} \label{Rts}
R_t - R_s = (t-s) \mu  + \int_s^t  \sqrt{V_u} dZ_u.
\end{equation}
Due to \eqref{mart6}, \eqref{mart7}, this implies for $0 \leq s \leq t \leq T$,
\begin{equation} \label{VtsRts}
\| V_t - V_s \|_q \leq \alpha_q (t-s) ^{1/2} 
\text{~~and~~} 
\| R_t - R_s \|_q \leq \delta_q (t-s) ^{1/2}  
\end{equation}
with 
$$
\delta_q = \sqrt{T} \mu + k_q
\text{~~and~~} 
\alpha_q =  \sqrt{T} \kappa (\theta + c_q) + \gamma k_q.
$$
This proves \eqref{HestonHolder1} for  $V_0 >0$ and $R_0$ fixed. 
A fully similar proof holds when the volatilities  $V_t$ are stationary, with only $R_0$ fixed. 
\qed
\end{proof}

\subsection{Realized Volatilities and actual volatilities} 
\label{RealVol}
Daily or Intraday market data provide observed asset prices $A_t$ at discretized times t, and hence discretized versions of the rate of return $R_t$, but the  spot variances $V_t$ cannot be directly observed or precisely derived from market data and hence are \emph{unobservable}. However \emph{observable approximations} of $V_t$ are provided for each small $\eps >0$ by the \emph{Realized Volatilities} $Y_t^\eps$ computed as follows from the discretized rates of returns $R_s$.

Partition each sliding time window $W_t^\eps = \left[ t- \eps, t \right]$ into $J(\eps)$ equal intervals, we define $J(\eps) +1$ time-instants
$$
t_n = t-\eps + n \eps / J(\eps) \text{~~for~~} n=0, \ldots, J(\eps).
$$
The realized volatilities $Y_t^\eps$ are then computed by the formula 
\begin{equation}
\label{RVformula}
Y_t^\eps = \frac{1}{\eps}\sum_{n=1}^{J(\eps)} (R_{t_n} - R_{t_{n-1}})^2.
\end{equation}
We will \emph{always assume} that the \emph{partition size} $J(\eps)$ verifies
$$
\lim_{\eps \to 0} J(\eps) = + \infty .
$$

\section{$L^q$-approximation of Heston Volatilities by Realized Volatilities}
\label{sec:Lqconverge}
As shown in \cite{Nielsen1}, when $\eps \to 0$, the realized volatilities $Y_t^\eps $ must converge to $V_t$ in $L^2$. For the Heston SDEs we now extend this result to convergence in $L^q$ for all $q \geq 2$, with estimates of the $L^q$ speeds of convergence.
\begin{mytho}
\label{Heston converge}
Fix any starting points $V_0 >0$ and $R_0$ for the returns rate $R_t$ and squared volatilities $V_t$ driven by Heston SDEs \eqref{HesEq1}, \eqref{HesEq2}.
Fix $T >0$ and any \emph{even} integer $q \geq 2$.
Then there is a constant $c(q)$ such that, for any $\eps >0$ and the choice of partition sizes $J(\eps)$, and any $t < T$, the realized volatilities $Y_t^\eps$ defined by \eqref{RVformula} verify
\begin{equation} \label{speedconv}
\|Y_t^\eps - V_t\|_q \leq c \left( \frac{1}{J^{1/q}(\eps)}+ \eps^{1/2} \right) .
\end{equation}
Hence, when $\eps \to 0$ and  $J(\eps) \to \infty$,  
the $Y_t^\eps$ converge to $V_t$ in $L^q$, uniformly over $0\leq t \leq T$.
Moreover if one imposes $J(\eps) > (a / \eps)^{q/2}$ for some fixed $a >0$, one has then
$$
\|Y_t^\eps - V_t\|_q \leq c \left(1+1/\sqrt{a} \right)  \eps^{1/2} .
$$
\end{mytho}
\begin{proof}
Recall a  well known lemma, easily proved by recurrence.
\begin{mylem} \label{ABC}
For each integer $q \geq 1$ and  any random variables $W_1, \ldots, W_q$ in $L^q$, one has
\begin{equation}
\label{EABC}
\big| \bE[W_1 W_2 \ldots W_q] \big| \leq \| W_1\|_q \| W_2 \|_q \ldots \|W_q\|_q.
\end{equation}
\end{mylem}
Fix $T >0, V_0 >0, R_0$. By \eqref{HestonHolder1} and \eqref{HestonHolder2}, for each $q \geq 1$, 
there is a constant $C_q$ such that for $0 \leq s  \leq t \leq T$ one has
\begin{equation}
\label{boundVR}
\| V_t \|_q + \| R_t \|_q \leq C_q \text{~~and~~} \| V_t - V_s\|_q + \| R_t - R_s\|_q \leq C_q (t- s)^{1/2} .
\end{equation}
Assume first that parameter $\mu=0$ in the
Heston SDE \eqref{HesEq1} for the returns.
Then  $R_t - R_s = \int_s^t \sqrt{V_u} dZu$, and by Ito formula, 
$$
(R_t - R_s)^2 = \int_s^t V_u du + 2 \int_s^t \; (R_u-R_s) \sqrt{V_u} \; dZ_u .
$$
Hence, the  variables $D(s,t)$ defined by 
\begin{equation} 
\label{D1}
D(s,t) = (R_t - R_s)^2 - \int_s^t V_u du  
\end{equation}
must verify 
\begin{equation} \label{D2}
\bE(D(s,t) \; | \; \calG_s) = 0 .
\end{equation}
Moreover \eqref{boundVR} gives 
\begin{equation}
\label{Rt-Rs}
\| (R_t -R_s)^2 \|_q =  \left( \| R_t -R_s \|_{2 q} \right)^2  \leq C_{2 q} (t-s)
\end{equation} 
and hence \eqref{D1} yields 
\begin{equation} \label{D3}
\| D(s,t) \|_q \leq \| (R_t -R_s)^2 \|_q + \int_s^t \| V_u \|_q du \leq b_q (t-s)
\end{equation}
with $ b_q = C_{2 q}^2 + C_q$.
For each $\eps > 0$, select the partition size $J(\eps)$, and partition the sliding window  $[t- \eps , t]$  by the time points  $t_n = t-\eps + n \eps/J$, with $n =0, \ldots, J$. Recall formula \eqref{RVformula} for the realized volatilities $Y_t^\eps$ 
$$
Y_t^\eps = \sum\limits_{n=1}^J \; (R_{t_n} - R_{ t_{n-1} })^2.
$$
To study $Y_t^\eps - V_t$ we introduce the decomposition 
\begin{equation}
\label{Y-V}
Y_t^\eps - V_t = H(t, \eps) + K(t, \eps) .
\end{equation}
where the  terms $H(t, \eps)$ and $K(t, \eps)$ are defined as
\begin{equation}
\label{HK}
H(t, \eps) = Y_t^\eps - \frac{1}{\eps}\int\limits_{t-\eps}^t V_u du 
\quad \text{and} \quad 
K(t, \eps) = \frac{1}{\eps}\int\limits_{t-\eps}^t V_u du - V_t .
\end{equation}
We can rewrite $K(t, \eps)$ as 
$$
K(t, \eps) = \frac{1}{\eps}\int\limits_{t-\eps}^t (V_u -V_t) du
$$
and this implies
$$
\| K(t, \eps) \|_q \leq \frac{1}{\eps}\int\limits_{t-\eps}^t \| V_t -V_u \|_q du .
$$
Equation \eqref{boundVR} gives the bound 
$\| V_t -V_u \|_q \leq C_q (t-u)^{1/2}$ 
for $ u \leq t \leq T$, and hence 
\begin{equation}
\label{boundKt}
\| K(t, \eps) \|_q \leq \frac{C_q}{\eps}\int\limits_{t-\eps}^t (t - u)^{1/2} du = \frac{2}{3} C_q \eps^{1/2}.
\end{equation}
We now study $H(t,\eps)$. Define $U_n$ for $n=1, \ldots, J$ by   
\begin{equation}
\label{Un}
U_n = D(t_{n-1} , t_n) = (R_{t_n} - R_{ t_{n-1} })^2 - \int\limits_{t_{n-1}}^{t_n} V_u du .
\end{equation}
Formula \eqref{HK} then implies directly
\begin{equation}
\label{HsumU}
H(t, \eps) = \frac{1}{\eps} \sum_{n=1}^{J} U_n.
\end{equation}
From \eqref{D3} and \eqref{Un} we obtain
\begin{equation}
\label{Unq}
\| U_n \|_q \leq b_q ( t_n - t_{n-1} ) = b_q \frac{\eps}{J}.
\end{equation}
Define  the polynomial $ Q = Q(U_1, \ldots, U_n) $ by
\begin{equation}
\label{defQ}
Q = Q(U_1, \ldots, U_n) = \left(\sum_{n=1}^{J} U_n \right)^q.
\end{equation}
Since $q$ is  \emph{even}, we then get, due to  \eqref{HsumU},  
\begin{equation}
\label{EHq}
| H(t, \eps) |^q = H(t, \eps) ^q = \frac{1}{\eps^q} \bE[Q]  . 
\end{equation}
Next, we derive the bound on $\bE[Q]$ with $Q$ defined in \eqref{defQ}.
To this end, we first define the set 
$M \equiv M(q, J)$ of all multi-integers 
$$
\vec{m} = (m(1), \ldots  , m(q) ) \;\; \text{with all} \;\; m(k) \in \left[ 1, J \right].
$$
For any $\vec{m} \in M$ denote by $Q_{\vec{m}}$ the monomial 
$Q_{\vec{m}} = U_{m(1)} U_{m(2)} \ldots U_{m(q)}$. Then, we can
expand  the polynomial $Q$ as follows
$$
Q = \left( \sum_{n=1}^{J} U_n \right)^q  = \sum_{\vec{m} \in M}  Q_{\vec{m}}.
$$
Then \eqref{EHq} is equivalent to 
\begin{equation}
\label{EEHq}
| H(t, \eps) |^q = \frac{1}{\eps} \sum_{\vec{m} \in M}  \bE[Q_{\vec{m}}]  .
\end{equation}
Due to lemma \ref{ABC} and \eqref{Unq}, we have for all $\vec{m} \in M$,
\begin{equation}
\label{EQm}
\big| \bE(Q_{\vec{m}}) \big| \leq \| U_{m(1)}\|_q \ldots \|U_{m(q)}\|_q \leq 
\left[\;\frac{b_q \eps}{J}\; \right]^q .
\end{equation}

The above bound is sufficient for most  $\vec{m} \in M$ but needs to be refined on a specific subset of $M$. 
So for each $\vec{m} \in M$, let $m^* = \max( m(1), \ldots , m(q) ) $.
Let $z(\vec{m})$ be the number of indices $m(k)$ such that $m(k) = m^*$.
For $1 \leq r \leq q$, call $M_r$ the set of $\vec{m} \in M$ such that $z(\vec{m})=r$. Then $M$ is the union of disjoint subsets $M_r$, $r=1,\ldots,q$.
Each $\vec{m} \in M_r$ contains at most $q - r+1$ distinct indices $m(k)$. Hence each $M_r$ has cardinal $Card(M_r) \leq J^{q - r+1}$. We now consider two cases separately (i) $r \ge 2$ and (ii) $r=1$.

For $r \geq 2$ we can have a an upper bound $Card(M_r) \leq J^{q-1}$, and therefore
$$
Card(M-M_1) = \sum_{r=2}^{q} Card(M_r) \leq (q-1) J^{q-1}.
$$
This yields, in view of \eqref{EQm}, 
\begin{equation}
\label{M-M1}
\left| \sum\limits_{\vec{m} \in (M- M_1)} \bE \left[Q_{\vec{m}} \right] \right| \leq  Card(M-M_1) \left[\; b_q \frac{\eps}{J} \;\right]^q \leq (q-1) b_q^q \eps^q / J.
\end{equation}

We now consider $M_1$ separately.
Fix any  $\vec{m} \in M_1$.
The maximum $m^* = \max (m(1) \ldots m(q))$ is then reached by a single index $i^*$ such that only $m(i^*) = m^*$ and $m(i) < m^*$ for $i \ne i^*$. This implies $2 \leq m^* \leq J$ since  $q \geq 2$. We can then re-order $\vec{m}$ 
as a multi-index $\vec{\nu}$ verifying 
$$
\nu(1) \leq \nu(2) \leq \ldots \leq \nu(q-1) \leq (m^*-1) < \nu(q) =m^*.
$$
Let $s=t_{j-1}$ and $t = t_j$.
For $1 \leq k \leq q-1$ all the $U_{\nu(k)}$ are $\calG_s$-measurable, so that 
$$
\bE [Q_{\vec{m}} | \calG_s] =  U_{\nu(1)} \ldots U_{\nu(q-1)} \bE[U_{\nu(q)} | \calG_s ].
$$
Due to \eqref{D2} and definition \eqref{Un}, we have 
$$
\bE \left[U_{\nu(q)} | \calG_s \right] = \bE \left[ D(s,t) |  \calG_s \right] =0
$$ 
and, therefore, $\bE [Q_{\vec{m}} | \calG_s] = 0$ for $\vec{m} \in M_1$.

Thus, one has
$$
\bE[Q] = \sum\limits_{\vec{m} \in M-M_1} \bE [Q_{\vec{m}}]
$$
and equation \eqref{M-M1} implies
\begin{equation}
\label{EQ}
\big| \bE[Q] \big| = \bE[Q] \leq (q-1) b_q^q \eps^q / J.
\end{equation}
Equations \eqref{EHq} and \eqref{EQ} then yield the bound
\begin{equation}
\label{boundHt}
\| H(t, \eps) \|_q = \frac{1}{\eps}\big| \bE [Q] \big|^{1/q} \leq 
3 b_q / J^{1/q} .
\end{equation}
Combining equation \eqref{Y-V} with the two bounds \eqref{boundHt} and \eqref{boundKt}, we conclude that  for all even $q \geq 2$, all $t \leq T$, and all $\eps >0$, one has 
\begin{equation}
\label{speed}
\| Y_t^{\eps} - V_t \|_q \leq 3 b_q / J(\eps)^{1/q} + C_q \eps^{1/2} .
\end{equation}
Hence when $\lim_{\eps \to 0} J(\eps) = +\infty$, we obtain the $L^q$ convergence  
$$
\lim_{\eps \to 0} \| Y_t^{\eps} - V_t \|_q = 0
$$
and this convergence is uniform for $0 \leq t \leq T$. To minimize the upper bound on the speed of convergence given by \eqref{speed}, one must clearly impose $J(\eps) \sim 1 / \eps^{q/2}$, and then as $\eps \to 0$ one has
$$
\| Y_t^{\eps} - V_t \|_q \sim \eps^{1/2}.
$$

For $\mu \ne 0$, the inequalities on the return process presented in the proof above still hold since we're considering a fixed time 
$T > 0$ and do not strive to have uniform bound in $T$. Consider, for instance, \eqref{Rt-Rs} with $\mu \ne 0$
\[
\| (R_t -R_s)^2 \|_q = \left( \| R_t -R_s \|_{2q} \right)^2 \le 
\left( \mu(t-s) + \left| \left| \int\limits_s^t \sqrt{V_u} dZ_u \right|\right|_{2q} \right)^2
\]
and using \eqref{mart7} we obtain the bound
\[
\| (R_t -R_s)^2 \|_q \le 
\left( \mu(t-s) + k_{2q} (t-s)^{1/2} \right)^2 \le \left( \mu^2 T + k_{2q}^2 + 2 \mu k_{2q} T^{1/2}\right) (t-s) \equiv
\tilde{C}_{T,q} (t-s).
\]
This concludes the proof.
\qed
\end{proof}
%

\section{Observable Estimators for the Heston Model Parameters }
\label{sec4}
\subsection{Parameter Estimation from true volatility data} To fit the Heston model to asset prices data, one needs to estimate the parameters $\mu, \rho$ of the SDE \eqref{HesEq1} and the parameter vector $\btheta$ of the Heston volatility SDE \eqref{HesEq2}. Since the volatility $V_t$ is unobservable, the key issue in parametric estimation of the Heston model is to estimate $\btheta$ (see, e.g. \cite{AG}).
Consider first the ideal but unrealistic case where we are given a large finite set of $N$ true volatilities values $\mathcal {V} = \{ V_{n\Delta} \} $, sub-sampled at time intervals $\Delta$.
For processes driven by smoothly parametrized SDEs, many publications have studied parameter estimation from large data sets \emph{actually generated} by the underlying SDEs 
(see for instance \cite{sahalia02,sahkim07,sahalia08,Basawa1,dusi93,genon1999parameter,gencat90,Phillips2009,Mariani2008,bates2006}, etc.). 
Several of these approaches rely either on Maximum Likelihood Estimators (MLEs) or on Methods of Moments.

\paragraph{Maximum Likelihood Estimators:} For the Heston volatility SDE, the MLEs of $\btheta$ have been thoroughly analyzed in \cite{AG}, where they are explicitly computed from any finite set of true squared volatilities  $\mathcal{V}$. Under minor parameter constraints and as $N \to \infty$, these MLEs were shown to be asymptotically consistent, and asymptotically normal when $\kappa \theta / \gamma^2 > 1$. 
Note that the impact of replacing the unobservable volatilities $V_t$ by the realized volatilities $Y_t^\eps$ in the explicit MLE formulas of \cite{AG} is a quite technical task which we will complete in a future paper. 

\paragraph{Moments based Estimators:} In this paper, we will focus instead on natural \emph{Moment Estimators} $\hat{\btheta}$ of the Heston SDE parameter vector $\btheta $. Since the true squared volatilities $V_t$ are \emph{unobservable}, $\hat{\btheta}$ is constructed as an explicit smooth function of the empirical mean and two lagged empirical covariances of the \emph{observable} realized volatilities $Y_t^\eps$.

\subsection{Parameter estimation under indirect observability}
The Moments Estimators approach considered in this paper falls formally within the generic \emph{Indirect Observability}
framework we introduced in \cite{azreti15}. In this framework, we analyze the \emph{observable} processes $Y_t^\eps $ which, as $\eps \to 0$, converge in $L^4$ to an \emph{unobservable} processes $X_t$ parametrized by a vector $\btheta$.
In particular, in \cite{azreti15}, 
after selecting a number of observables $N(\eps)$ and a sub-sampling rate $\Delta(\eps)$, 
the \emph{observable estimators} $\hat{\btheta}(\eps)$ of $\btheta$ are constructed as smooth functions of the empirical mean and a finite set of empirical lagged covariances of the observables $Y_{n \Delta(\eps)} ^\eps , 1 \leq n \leq N(\eps)$.
Under a broadly applicable set of \emph{Indirect Observability Hypotheses}, which however require $X_t$ to be weakly stationary, we proved in \cite{azreti15} that one can construct observable moment estimators achieving consistency as $\eps \to 0$, provided $N(\eps)$ and $\Delta(\eps)$ are adequately selected.

Here we focus on the following indirect observability situation: \\
(i) the unobservable process $X_t$ is the squared volatility process $V_t$ \\
(ii) the observable $Y_t^\eps$ converging to $V_t$ as $\eps \to 0$ are the realized volatilities defined by the rate of returns process associated to $V_t$.

Note that in the present paper the volatility process $V_t$ starting at a deterministic $V_0 = y >0$ is \emph{not stationary}; therefore, the analytical results of the present paper have requires several quite technical enhancements of the methods previously used in \cite{azreti15}.

\section{Transition Densities for squared volatilities}
\label{sec5}
\subsection{Explicit transition density}
Consider squared volatilities $V_t$ driven by the Heston volatility SDE \eqref{HesEq2} parametrized by $\btheta$. We will always assume that $V_0 >0$ has finite moments of all orders, which is of course the case if $V_0$ is deterministic.
For $T>0 $, introduce the following short-hand notations 
\begin{equation} \label{notations}
\nu_T = e^{- \kappa T}, \qquad r = \frac{2\kappa\theta}{\gamma^2}-1 > 0, \qquad
\Lambda = \frac{2\kappa}{\gamma^2}, \qquad \lambda_T = \frac{\Lambda}{1 - \nu_T}.
\end{equation}
As shown in \cite{Cox}, the Markov diffusion process $V_t$ has an explicit transition density $p_T(z,y)$, which we often denote $p(z,y)$ for short, given by
\begin{equation} \label{cox1985theory}
p_T(z, y) = P(V_{s+T}= z \, | \,V_s=y ) =
\lambda_T \left( \frac{z}{y \nu_T} \right)^{r/2} \exp \left(- \lambda_T (z + y \nu_T) \right) \, 
I_r \left( 2 \lambda_T \sqrt{z y \nu_T } \right), 
\end{equation}
where $I_r$ is the modified Bessel function of the 1st kind of order $r$.
As noted in \cite{Cox}, for fixed $T$, the linear rescaling $V_t \to 2 \lambda_T V_t$ transforms the transition density $p_T(z, y)$ into 
\begin{eqnarray}
P ( 2 \lambda_T V_T = z \, | \, 2 \lambda_T V_0 = y) 
= p_T \left(\frac{z}{2 \lambda_T}, \frac{y}{2 \lambda_T} \right) \frac{1}{2 \lambda_T} = 
\nonumber \\ 
\frac{1}{2} \left( \frac{z} {y \nu_T} \right)^{r/2} \exp( - (z + y \nu_T)/2) \, I_r ( \sqrt{z y \nu_T } ),
\nonumber 
\end{eqnarray}
which, for each fixed $y$, is a \emph{non-central $\chi^2$ density} with non-centrality parameter 
$NCP(T,y) = y \nu_T$ and $DFR = 2r + 2$ degrees of freedom. 
Note that $DFR = {4 \kappa \theta}/{\gamma^2}= 2 \theta \Lambda$.

\subsection{The stationary squared volatility process $V_t$}
\label{stationaryVt}
Since $\nu_T \to 0$ and $\lambda_T \to \Lambda $ as $T \to \infty$, $p_T(z,y)$ converges pointwise, at the speed $e^{- \kappa T}$, to the \emph{unique stationary density} 
$\psi(z)$ of the autonomous Heston volatility SDE. This stationary density is given for all $z > 0 $ by 
\begin{equation} \label{cox}
\psi (z)= \frac{\Lambda}{\Gamma(r+1)}(\Lambda z)^{r} \exp( - \Lambda z ).
\end{equation}
When the initial condition $V_0$ is random with density $\psi$, 
all $V_t$ have then the same density $\psi$, and the process $V_t$ driven by the Heston volatility SDE becomes \emph{strictly stationary}. Expectations with respect to this stationary diffusion will be denoted $\bE_{\psi}$. 

Note, that since $\lim_{T \to \infty} NCP (T,y) = 0$, the linear rescaling $z \to 2 \Lambda z $ transforms the stationary density $\psi(z)$ 
into $({2 \Lambda})^{-1} \psi(z / 2 \Lambda)$ which is a \emph{standard $\chi^2$} density with $DFR = 2r + 2$ degrees of freedom.

\section{Conditional Moments of squared volatilities}
\label{sec:moments}
\subsection{Moments of non-central $\chi^2$} 
\label{sec6.1}
Let $X$ be a random variable having a non-central $\chi^2$ density with $DFR$ degrees of freedom and non-centrality parameter $NCP$. Then, the Laplace transform $Lap(z) = \bE(e^{z X})$ is, for $ 0 \leq z < 1/2$, 
$$
Lap(z) = (1-2z)^{DFR}\exp \left( NCP \frac{z}{1-2z} \right) = (1-2z)^{DFR}\sum_{n=0}^{\infty} \frac{1}{n!}(NCP)^n 
\left( \frac{z}{1-2z} \right)^{n}.
$$ 
Expanding $1/(1-2z)^n$ as a series in $z$, we obtain, for a fixed $DFR$, 
$$
Lap(z) = \sum_{q=0}^{\infty} \pi_q(NCP) \frac{z^q}{q!}, 
$$
where the $\pi_q(NCP)$ are \emph{polynomials} of degree $q$ in $NCP$, with coefficients fully determined by $DFR$ and $q$. 
Denote $nc\chi(q)$ and $st\chi(q)$ the respective moments of order $q$ for the non-central $\chi^2$ density and for the standard $\chi^2$ density with $DFR$ degrees of freedom. We then have the polynomial expressions 
\begin{equation}
\label{Mnoncentral}
nc\chi(q) = \pi_q(NCP) \;\;\; \text{and} \;\;\; st\chi(q) = \pi_q(0).
\end{equation}
For the first two moments of the non-central $\chi^2$ and the standard $\chi^2$, one has, for instance, the following well known formulas
\begin{equation}
\label{firstmoments}
\begin{aligned}
nc\chi(1) &= \pi_1(NCP) = DFR + NCP, \\
nc\chi(2) &= \pi_2(NCP) = NCP^2 + 2 NCP (DFR + 2) + DFR^2 + 2 DFR, \\ 
st\chi(1) &= DFR, \\ 
st\chi(2) &= DFR^2 + 2 DFR. \\
\end{aligned}
\end{equation}

\subsection{Conditional Moments of the squared volatilities}
Recall that $\nu_T = e^{-\kappa T}$ and that $DFR = 2r+2$ is determined by $\btheta$. 
The next proposition addresses the computation of conditional moments for $V_t$
\begin{equation}
\label{Ey}
M_q(y, T) \equiv \bE[V_{s+T}^q \; | \; V_s = y] = \bE[V_T^q \; | \; V_0 = y] \equiv \bE_y[V_T^q].
\end{equation}
\begin{mypro} \label{cond.moments}
For each $q \geq 1$ and for each $y >0$, the conditional moments $\bE_y[V_T^q]$ remain uniformly bounded for all $T \geq 0$.
There is a polynomial $Q_q$ with coefficients depending only on $q$ and $\btheta$, 
such that for all $s$, $T$ and all $y >0$,
\begin{equation}\label{Mpol}
M_q(y, T) \equiv \bE[V_{s+T}^q \; | \; V_s = y] = Q_q(y).
\end{equation}
As $T \to \infty$, moments $ M_q(y, T)$ converge at the exponential speed $\nu_T = e^{- \kappa T}$ 
to finite moments 
$m_q = \bE_\psi [V_t^q]$ of the stationary diffusion $V_t$ driven by the Heston volatility SDE. 
\end{mypro}

\begin{proof}
The rescaling $V_s \to 2 \lambda_T V_s$, with $\lambda_T$ in \eqref{notations}, transforms the conditional distribution of $V_{s+T}$ given that $V_s = y$ into a non-central $\chi^2$ with 
$$
DFR = 2r +2 = 2 \theta \Lambda \; \; \; \text{and} \; \; \;
NCP = (2 \lambda_T y) \nu_T = \nu_T \frac{2\Lambda}{1-\nu_T} y.
$$ 
This rescaling implies, using the non-central $\chi^2$ moments \eqref{Mnoncentral},
\[
\begin{aligned}
M_q(y, T) &= \frac{1}{(2\lambda_T)^q} \bE \big[ (2\lambda_T V_T)^q \; | \; 2\lambda_T V_0 = 2 \lambda_T y \big] = 
\frac{1}{(2\lambda_T)^q} \pi_q \big( 2\lambda_T y \nu_T \big) \\
&= \left[ \frac{1-\nu_T} {2\Lambda}\right]^q \pi_q \left( y \nu_T \frac{2\Lambda}{1-\nu_T} \right) .
\end{aligned}
\]
Define the homogeneous polynomial $H_q(a, b)$ of total degree $q$ by 
\begin{equation}
\label{Hq}
H_q (a, b) = a^q \pi_q \big( b/a \big), 
\end{equation}
where $\pi_{q}(\cdot)$ is defined by the Laplace transform introduced in section \ref{sec6.1}.
Therefore, coefficients of $H_q$ depend only on $q$ and on $\btheta$. Next, we define 
\begin{equation}
\label{Q.pol}
Q_q(y) = H_q \left( \frac{1-\nu_T}{2\Lambda}, y \nu_T \right).
\end{equation}
Since $0 \leq \nu_T \leq 1$ is a constant given by \eqref{notations}, 
the expression \eqref{Q.pol} for $Q_q(y)$ proves equation \eqref{Mpol}.
Equation \eqref{Mpol} also implies the existence of a constant $C$ such that 
$$
M_q(y,T) \leq C (1 + y)^q \;\; \text{for all} \; \; T \geq 0.
$$ 
Therefore for each $V_0 =y >0$, moments $\bE_y[V_T^q]$ remain bounded for all $T \geq 0$.
As discussed in section \ref{stationaryVt}, rescaling by $\Lambda$ transforms the stationary density $\psi$ into a standard $\chi^2$ distribution, and hence stationary moments 
$$
m_q = \bE_\psi[ V_T^q ] = \int_{y \geq 0} y^q \psi(y) dy 
$$ 
must be finite.
As $T \to \infty$, both $\nu_T= e^{- \kappa T}$ and $NCP$ tend to $0$, and $\lambda_t \to \Lambda$ while DFR remains constant. Hence, due to \eqref{Mpol}, \eqref{Q.pol} the $ M_q(y, T)$ converge at exponential speed $\nu_T$ to 
$$
m_q = \bE_{\psi} [V_s^q] = H_q \left(\frac{1}{2\Lambda}, 0 \right) \equiv \pi_q(0).
$$

\qed
\end{proof}

\subsection{Mean and Covariances of the stationary diffusion $V_t$} \label{covarVt}
Using \eqref{firstmoments}, and the appropriate rescaling of $V_t$ by $2 \lambda_T$ one easily computes the first two conditional moments of the squared volatility process starting at $y>0$, namely
\begin{eqnarray}
M_1(y,T) &=& \bE_y [V_T] = (1-\nu_T) \theta + \nu_T y , 
\label{M1} \\
\nonumber \\ 
M_2(y,T) &=& \bE_y \left[V_T^2 \right] = y^2 \nu_T^2 + 2 y \nu_T (1-\nu_T)(\theta + 1/\Lambda) +(1-\nu_T)^2 
\theta (\theta +1 / \Lambda) , 
\label{M2}
\end{eqnarray}
where $\bE_y [\cdot]$ is the conditional moment defined in \eqref{Ey}.
In particular, as $T \to \infty$, equations \eqref{M1} and \eqref{M2} yield the first two moments of the stationary diffusion $V_t$
\begin{equation} \label{m1m2}
m_1 = \bE_\psi [V_s ] = \theta \quad \text{and} \quad m_2 = \bE_{\psi} \left[ V_s^2 \right] = \theta^2 + \theta / \Lambda.
\end{equation}
Moreover, the stationary diffusion driven by the Heston volatility SDE has mean $m_1 = \theta$ and lagged covariances $K(u) = cov_{\psi} [ V_s V_{s+u} ]$ given by
\[
K(u) + \theta^2= \bE_{\psi} [ V_s V_{s+u}] = \bE_{\psi} \left[ V_s M_1(V_s, u) \right] = \bE_{\psi}\left[ V_s (1-\nu_u) \theta + \nu_u V_s \right] = \theta^2 + \nu_u (m_2 - \theta^2)
\]
for any time lag $u \geq 0$.
This yields the stationary covariances
\begin{equation} \label{covstationary}
K(u) = e^{- u \kappa} \frac{\theta \gamma^2}{2 \kappa} \quad \text{and} \quad 
K(0) = \frac{\theta \gamma^2}{2 \kappa}.
\end{equation}

\subsection{Heston SDE parameters as functions of asymptotic moments}
\label{bthetaformula}
We can now express $\btheta = (\kappa, \theta, \gamma)$ as an explicit \emph{smooth} function 
$$
\btheta = \Phi \left( m_1, K(0), K(u) \right)
$$
of three moments of the \emph{stationary} volatility diffusion $V_t$, namely its mean $m_1$, its variance $K(0)$, and one lagged covariance $K(u)$ for some fixed (but arbitrary) $u >0$. Equations \eqref{m1m2} and \eqref{covstationary} indeed imply that parameters $(\kappa, \theta, \gamma)$ can be expressed using the moments $m_1$, $K(0)$, and $K(u)$ as follows
\begin{equation} \label{par3}
\theta = m_1 = \bE_{\psi} [V_t] \; , \qquad
\kappa = - \frac{1}{u} \log\left(\frac{K(u)}{K(0)}\right) , \qquad 
\gamma^2 = \frac{2 K(0) \kappa}{ \theta},
\end{equation}
which defines the function $\Phi$ above.

\section{Moments based observable estimators}
\label{sec7}
We now use our preceding results on the Heston volatility SDEs 
to study a class of moment-based estimators of the Heston parameters and to discuss their consistency when the observable data are generated by the realized volatilities.

\subsection{Computation of Moments Based Observable Estimators} 

Given the window-size $\eps>0$, select a sub-sampling time interval $\Delta \equiv \Delta(\eps) $ and a number of observations $N \equiv N(\eps)$.
Then, the realized volatility process \eqref{RVformula} generates an observable data set of $N(\eps)$ realized volatilities 
$$
W_k= Y^\eps_{k \Delta(\eps)}, \;\; k=1 \ldots N(\eps).
$$
Next, we specify how we use these $N(\eps)$ observable data to estimate any lagged covariance $K(u)$ of the stationary diffusion $V_t$. 
Denoting $\left[ a \right]_{int}$ the closest integer to $a$, we approximate the lag $u$ by $U \Delta(\eps)$ where 
\begin{equation}\label{timelagU}
U = U(u, \eps) = \left[ \frac{u}{\Delta(\eps)} \right]_{int} \;\; \text{so that} \;\; |u - U \Delta(\eps) | \leq \Delta(\eps).
\end{equation}
Since $K(u)$ is Lipschitz in $u$, there is a constant $C \equiv C(u)$ such that
$$
|K(u) - K(U \Delta(\eps)) | \leq C \Delta(\eps) \;\; \text{for all} \;\; \eps >0.
$$ 
Then, for any $\eps$ and time lag $u$, we define observable empirical estimators of the mean $m_1$ and lagged covariances 
$K(u)$ of the stationary diffusion $V_t$ as follows
\begin{equation} \label{empirical}
\hat{m}^\eps = \frac{1}{N}\sum\limits_{k=1}^{N} W_k , \quad
\hat{K}(u) \equiv \hat{K}^\eps(u) = - (\hat{m}^\eps)^2 + \frac{1}{N-U} \sum\limits_{k=1}^{N-U} W_k W_{k+U} , 
\end{equation}
where $ U = U(u, \eps) $ and $N = N(\eps)$.
Formulas \eqref{par3} express the parameter vector $\btheta$ as an explicit $C^1$ function $\Phi(m_1, K(0), K(u))$. This naturally leads to the definition of an observable parameter estimator $\hat{\btheta}(\eps) $ of $\btheta$ by
$$
\hat{\btheta}(\eps) = \Phi(\hat{m}^\eps, \hat{K}^\eps (0), \hat{K}^\eps (u)).
$$
This definition yields the following explicit observable estimators of the Heston parameters
\begin{equation}\label{estim}
\hat{\theta}(\eps) = \hat{m}^\eps , \qquad
\hat{\kappa}(\eps) = 
-\frac{1}{u}\log\left( \frac{\hat{K}^\eps(u)}{\hat{K}^\eps(0)} \right) , \qquad
\hat{\gamma}^2 (\eps)= \frac{2 \hat{K}^\eps(0) \hat{\kappa}(\eps)}{\hat{\theta}(\eps)}.
\end{equation}
\subsection{Asymptotics for polynomial functionals of squared volatilities}
\begin{mytho}\label{polymoments}
Consider any fixed polynomial $h(x_1, \ldots, x_k)$ of total degree $n$ in $k$ variables $(x_1, \ldots, x_k)$. 
Let $0 = u(0) < u(1) < \ldots < u(k)$ be any sequence of $k+1$ lag instants.
For $T >0$, define $H$ and $H_T$ by
$$
H= h \left( V_{u(1)}, \ldots, V_{u(k)} \right) \;\;\; \text{and} \;\;\;
H_T= h \left( V_{u(1)+T}, \ldots, V_{u(k)+T} \right) .
$$
Recall that $\nu_T = e^{- \kappa T}$. Define 
$w_j = e^{ - \kappa (u(j+1) - u(j))}$ for $j = 0, \ldots, k-1$. Then, there is a polynomial $POL$ in $k+2$ variables such that for all $T>0$ and all $y>0$
\begin{equation}\label{EHT}
\bE_y (H_T) = POL(\nu_T, y \nu_T, w_0, w_1, \ldots , w_{k-1}) .
\end{equation}
The degree and coefficients of POL are determined by the integers $n$, $k$, the coefficients of $h$, 
and the vector $\btheta$. 
The asymptotic polynomial moments are then given by
$$
\lim_{T \to \infty} \bE_y (H_T) = \bE_\psi (H) = POL(0, 0, w_0, w_1, \ldots , w_{k-1}).
$$
For any integer $q \geq 1$ there is a positive constant $C$, and an integer $p \geq 1$, determined only by 
$q$, $k$, $\btheta$ and the polynomial $h$ such that, for all positive $T$ and $y$, and all $0 = u(0) < u(1) < \ldots < u(k)$
\begin{equation}\label{Lqlimitpol}
\Big| \bE_y \big[ \left( H_T - \bE_\psi (H) \right)^q \big] \Big| \leq C (1 + y^p) e^{- \kappa T} .
\end{equation}
In particular for $q= 1$ one has
\begin{equation}\label{limitpol}
\Big| \bE_y (H_T) - \bE_\psi (H) \Big| \leq C (1 + y^p) e^{- \kappa T}.
\end{equation}
\end{mytho}
\begin{proof}
For better readability, the detailed proof is given in Appendix \ref{polyfunctionals}.
\end{proof}

\noindent
\textbf{Remarks.} Equation \eqref{Lqlimitpol} also implies that as $T \to \infty$, the random polynomial functions $H_T$ converge in $L^q$-norm to the constants $\bE_\psi(H)$, where $L^q$-norms are computed under $\bE_y$. Note that the constant $C$ introduced in the theorem does not depend on the time lags $u(0) < u(1) < \ldots < u(k)$.

\subsection{Consistency of observable estimators}
\label{sec:consist}
Since $\hat{\btheta}(\eps)$ is a $C^1$ function of three specific empirical moments of realized volatilities, 
the key consistency issue is to estimate, as $\eps \to 0$, the speeds of convergence of $\hat{m}^\eps$ to $m_1$ and $\hat{K}^\eps (u)$ to $K(u)$.
These speeds of convergence strongly depend on the sub-sampling scheme defined by $N(\eps)$ and $\Delta(\eps)$. In \cite{azreti15}, we have determined sub-sampling schemes optimizing these speeds of convergence for \emph{stationary} unobservable limit processes.
We now prove similar results for the \emph{non-stationary} volatilities driven by the Heston SDEs.
\begin{mytho} \label{consistency} 
Consider an asset with return rate $R_t$ and squared volatility $V_t$, jointly driven by the Heston SDEs \eqref{HesEq1}, \eqref{HesEq2}. Fix deterministic initial conditions $R_0$ and $V_0 = y >0$. Call $P_y$ the probability distribution in path space of the trajectories $\{R_t, V_t\}$.
Realized volatilities $Y_t^{\eps}$ are computed by formula \eqref{RVformula} with $J(\eps) \sim \eps^{-2}$.
The $Y_t^{\eps}$ are sub-sampled with time step $\Delta(\eps)$ to generate $N(\eps)$ observations $W_k=Y_{k\Delta(\eps)}^\eps$. 
We then apply formulas \eqref{empirical} to compute observable empirical estimators $\hat{K}^\eps(u)$ and $\hat{m}^\eps$ of the asymptotic lagged covariances $K(u) - \theta^2 = \lim_{t\to \infty} \bE [V_t V_{t+u}]$ and mean $m_1= \lim_{t\to \infty} \bE [V_t]$ of true volatilities.

Then
there exists a sub-sampling scheme which
guarantees that for any fixed positive $L$ and $y$ there is a constant $C=C(L, y,\btheta)$ such that, for all lags $0 \leq u \leq L$, one has 
\begin{equation}
\label{speedCovar}
\| \hat{K}^\eps (u) - K(u) \|_2 \leq C \eps^{1/2} \;\;\; \text{and} \;\;\;
\| \hat{m}^\eps - m_1 \|_2 \leq C \eps ^{1/2}, 
\end{equation}
where $L^2$-norms are computed with respect to $P_y$. Moreover, under $P_y$ the observable parameters estimators $\hat{\btheta}^\eps$ given by formulas \eqref{estim} converge in probability to the true Heston parameters $\btheta$ as $\eps \to 0$. One has, for an adequate constant $C$,
\begin{equation}
\label{convProba}
P_y \left( \| \hat{\btheta}^\eps - \btheta \|_{\bR^3} \geq \eps^{1/3} \right) \leq C \eps^{1/3}.
\end{equation}
\end{mytho} 
\begin{proof}
Fix the time lag $u$ and $V_0 = y >0$. All $L^q$-norms are computed under $P_y$. The notation ``constant $C$'' will designate a generic constant which can change values from one bound to another.
By Theorem \ref{Heston converge}, there is a constant $c_4$ such that for all $t$, 
\begin{equation}
\label{bndbase}
\| V_t \|_4 \leq c_4 \;\; \text{and} \;\;
\| V_t - Y_t^\eps \|_2 \leq \| V_t - Y_t^\eps \|_4 \leq c_4 \eps^{1/2}.
\end{equation}
Therefore, there is a constant $c_2$ such that for all $s$ and $t$,
$$
\| V_s V_t - Y_s^\eps Y_t^\eps \|_2 \leq c_2 \eps^{1/2}.
$$
Denote $\Delta \equiv \Delta(\eps)$. 
The sub-sampled  realized volatilities $W_k=Y_{k\Delta}^\eps$ determine the observable empirical estimators of 1st and 2nd moments of volatilities, through formula \eqref{empirical}. 
Since $\| V_{k \Delta} - W_k \|_2 \leq c_4 \eps^{1/2}$ by \eqref{bndbase}, the definition \eqref{empirical} of $\hat{m}^\eps $ gives 
\begin{equation}
\label{bnd0}
\| \hat{m}^\eps - m_1 \|_2 \leq \frac{1}{N} \sum_{k=1}^{N} \left|\left| W_k - V_{k \Delta} \right|\right|_2 + 
\left| \left| \frac{1}{N} \sum_{k=1}^{N} V_{k \Delta} - m_1 \right|\right|_2 
\leq c_4 \eps^{1/2} + \frac{c_5}{\sqrt{N\Delta}},
\end{equation}
where we used \eqref{covstationary} which implies that 
\[
\sum\limits_{j=1}^\infty (V_{k \Delta} - m_1)  (V_{(k+j) \Delta} - m_1) \le Const < \infty.
\]
We would like to point out that the term $O(1/\sqrt{N\Delta})$ 
in the expression above arises from the $L^2$ error of the empirical mean, $m_1$, computed form direct observations, $V_{k\Delta}$. Therefore, the estimate in \eqref{bnd0} cannot be improved analytically.
 
Provided we take $N \Delta=\eps^{-1}$,
this proves convergence, at speed $\eps^{1/2}$, of the empirical mean of realized volatilities $\hat{m}^\eps$ to the asymptotic  mean of $V_t$. 
Next, we study our estimators of lagged covariances.
Let $U$ be the closest integer to $\left[ u / \Delta(\eps) \right]_{int}$,  so that $| U \Delta - u \Delta | \leq \Delta$ . Define
$$
H(W) = \hat{K}^\eps(u) + \left( \hat{m}^\eps \right)^2 = \frac{1}{N} \sum_{k=1}^{N} W_kW_{k+U} 
\;\;\; \text{and} \;\;\;
M(V) = \frac{1}{N} \sum_{k=1}^{N} V_{k \Delta} V_{(k+U) \Delta}.
$$
By subadditivity of norms, we obtain 
\begin{equation}
\label{bnd00}
\| H(W) - M(V) \|_2 \leq c_2 \eps^{1/2}.
\end{equation}
Let $\psi$ be the \emph{asymptotic} density of $V_t$, and denote by 
$$
m_2(s) = \bE_\psi [ V_t V_{t+s} ]
$$ 
the stationary lagged 2nd moments, which do not depend on $t$. By Theorem \ref{polymoments}, for every fixed $a$ and $y >0$ there is a constant $C$ such that for all $T$ and all $s\leq a$ one has the bound
\begin{equation}
\label{bnd1}
\bE_y \left[ (V_T V_{T+s} - m_2(s))^2 \right] \leq C e^{- \kappa T}.
\end{equation}
The $L^2$ norm under $P_y$ hence verifies 
\begin{equation}
\label{bnd2}
\| V_{j \Delta} V_{j \Delta + U} - m_2(U) \|_ 2 \leq C e^{- \frac{\kappa}{2} j \Delta}. 
\end{equation}
The above two bounds in \eqref{bnd1} and \eqref{bnd2} 
provide constants $C$ and $c= \kappa/2$ such that for all $\eps$
\begin{equation}
\label{ineq1}
\sum_{j =1}^N \| V_{j \Delta} V_{j \Delta + U} - m_2(U) \|_ 2 \leq C \frac{ e^{- c \Delta} }{1 - e^{- c \Delta}}
\leq \frac{ C}{ c \Delta}.
\end{equation}
By subadditivity of $L^2$ norms, inequality \eqref{ineq1} then implies, 
\begin{equation}
\label{bnd3}
\| M(V) - m_2(U) \|_2 \leq \frac{C}{c N\Delta}.
\end{equation}
Regrouping our definitions and notations, we have 
$$
\hat{K}^\eps(u) = - (\hat{m}^\eps)^2 +H(W), \quad 
K(u) = - m_1^2 + m_2(u), \quad 
K(U) = -m_1^2 + m_2(U).
$$
This implies since $K(u)$ is Lipschitz in $u$, 
\begin{equation}
\label{bnd4}
| m_2(U) - m_2(u) | = | K(U) - K(u) | \leq C |U-u | \leq C \Delta.
\end{equation}
We have the obvious identity
$$
\hat{K}^\eps(u) - K(u) =H(W) - (\hat{m}^\eps)^2 - \left( m_2(u) - m_1^2 \right) - M(V) +M(V) - m_2(U) + m_2(U)
$$
and hence 
$$
\| \hat{K}^\eps(u) - K(u) \|_2 \leq \| m_1^2 - (\hat{m}^\eps)^2 \|_2 + \| H(W) - M(V) \|_2 + \| M(V) - m_2(U) \|_2 + |m_2(U) - m_2(u)|.
$$
We now use the bounds \eqref{bnd0}, \eqref{bnd00}, \eqref{bnd3}, and \eqref{bnd4} to obtain 
\begin{equation}
\label{L2Kufinal}
\| \hat{K}^\eps(u) - K(u) \|_2 \leq c_4 \eps^{1/2} +\frac{c_5}{\sqrt{N\Delta}}+
c_2 \eps^{1/2} + \frac{C}{\sqrt{N}} + \frac{C}{N\Delta} + C \Delta.
\end{equation}
To optimize this last bound and ensure that all terms have the same rate of convergence as $\eps \to 0$, we impose the choice 
\begin{equation}
\label{optimal1}
\frac{1}{\sqrt{N\Delta}} \sim \Delta \sim \eps^{1/2}
\end{equation}
which is equivalent to selecting $\Delta \sim \eps^{1/2}$ and $N \sim \eps^{-3/2}$. 

Therefore, for each fixed $V_0=y >0$ and for all time lags $u$ within any fixed interval $\left[ 0, L \right]$ there is a constant $C$ realizing the bound 
$$
\| \hat{K}^\eps(u) - K(u) \|_2 \leq C \eps^{1/2}.
$$
The $L^2$ convergence under $P_y$ of $\hat{K}^\eps(u)$ to $K(u)$ and of $\hat{m}^\eps$ to $m_1$, implies their convergence in probability under $P_y$. By equation \eqref{estim} our estimators of Heston parameters are of the form 
$$
\hat{\btheta}^\eps = \Phi(\hat{K}^\eps(0), \hat{K}^\eps(u), \hat{m}^\eps)
$$
where $\Phi$ is a $C^1$ function. Thus, estimators $\hat{\btheta}^\eps$ converge in probability to 
$\btheta =$ $\Phi(K(0), K(u), m_1)$ as $\eps \to 0$.
The $L^2$-speeds of convergence $\eps^{1/2}$ for the 1st and 2nd moments imply, by Chebyshev inequality,
$$
P_y \left( |\hat{K}^\eps(u) - K(u)| \geq \eps^{1/3} \right) \leq C \eps^{1/3}
$$
with a similar inequality for $\hat{m}^\eps$. 
Since $\Phi$ is $C^1$, these speeds of convergence in probability under $P_y$ imply, by the first order Taylor expansion of the function $\Phi$, the same speed of convergence in probability for the parameter estimators themselves.

\qed
\end{proof}

\section{Asymptotically Optimal Partition Sizes $J(\eps)$}
\label{sec8}

Denote by $T(\eps)$ the \emph{total observation time} available for the rate of returns process $R_t$ The realized volatilities $Y^\eps_t$ given by formula \eqref{RVformula} involve subdividing the sliding window $(t -\eps, t)$ into $J(\eps)$ intervals and averaging the corresponding $J(\eps)$ squared increments of the returns rate. To compute our 1st and 2nd moments estimators from the observable process $Y^\eps_t$ with $t \leq T(\eps)$, we sub-sample this  process at the $N(\eps)$ instants $k \Delta(\eps)$, with $k = 1, \ldots, N(\eps)$, with the obvious relation 
$N(\eps) \Delta(\eps) = T(\eps)$.

Provided one uses the subsampling scheme
\begin{equation}\label{subs} 
N(\eps) \sim \eps^{-3/2}, \qquad \Delta(\eps) \sim \eps^{1/2}
\end{equation}
and a partition size $J(\eps) \sim 1/\eps^2$, our current theoretical bounds can guarantee $L^4$ speeds of convergence $\sim \sqrt{\eps}$ for $Y^\eps_t - V_t$ and consistency in probability of our observable estimators for the parameters in the volatility Heston SDE. 
The theoretical choice $J(\eps) \sim 1/\eps^2$ seems overwhelmingly large for concrete fitting of Heston SDEs to actual stockprices data. 
Therefore, we also examine numerically a
more pragmatic choice $ J(\eps) \sim \eps^{-1}$.

\noindent
\textbf{Remark.} The optimized sub-sampling scheme \eqref{subs} necessitates a total observational time $T(\eps) \sim \eps^{-1}$. The associated convergence rate \eqref{speedCovar} for the 1st and 2nd moments estimators can hence be reformulated as
\[
\| \hat{K}^\eps (u) - K(u) \|_2 \leq \frac{C}{\sqrt{T(\eps)}} \;\;\; \text{and} \;\;\;
\| \hat{m}^\eps - m_1 \|_2 \leq \frac{C}{\sqrt{T(\eps)}}.
\]
However, to compute realized volatilities $Y_t^{\eps}$, we need $J(\eps) \sim 1 /\eps^{2}$ time points in each  sliding  window, and hence 
the computation of the observable moments estimators requires a total number of observational points 
$$
n \equiv n(\eps) =N(\eps) J(\eps) \sim 1/ \eps^{7/2}.
$$ 
Therefore, the convergence rates of our moments estimators can be expressed as follows in terms of the total number $n$ of observational time points for the return process as
\[
\| \hat{K}^\eps (u) - K(u) \|_2 \leq C  n^{-1/7} \;\;\; \text{and} \;\;\;
\| \hat{m}^\eps - m_1 \|_2 \leq C  n^{-1/7}.
\]
Even the more pragmatic choice  $J(\eps) =1 / \eps$ results in the scaling
\[
\| \hat{K}^\eps (u) - K(u) \|_2 \leq C  n^{-1/5} \;\;\; \text{and} \;\;\;
\| \hat{m}^\eps - m_1 \|_2 \leq C  n^{-1/5}.
\]
These upper bounds on convergence rates are suboptimal for the somewhat theoretical case when the available observational time is not a priori bounded. Indeed Hoffmann \cite{hoffmann2002} indicates that when the total observational time $T \to\infty$  one should expect a more classical convergence rate $n(\eps)^{-1/2}$, while for fixed finite $T$, the paper \cite{hoffmann2002} suggests that the  optimal convergence rate should be $n(\eps)^{-1/4}$. However, results in \cite{hoffmann2002} focus on approximate maximum likelihood estimators, and, therefore, cannot be directly applied to the moments based estimators which are considered in this paper.

%

\section{Effective Speed of Convergence of realized volatilities to true volatilities}
\label{speedRealVol}
\subsection{Generic stochastic volatility models versus Heston SDEs}
\label{GENvsHES}
Recall that realized volatilities $ Y_t^\eps $ are computed by formula \eqref{RVformula} with partition size $J(\eps)$ for the time windows $\left[ t - \eps , t \right]$.
When the rates of return $R_t$ and the squared volatilities $V_t$ are driven by joint Heston SDEs, we have proved in Theorem \ref{Heston converge} that for each fixed even integer $q$ and for $s$ bounded, 
the $L^q$ norms $\|Y_s^\eps - V_s\|_q $ must verify, for some constant $C= C(q)$, the bounds
\begin{equation}
\label{L2L4bound}
\|Y_s^\eps - V_s\|_q \leq C \left( 1/J(\eps)^{1/q} + \sqrt{\eps} \right) .
\end{equation}
Our numerical simulations suggest that for the
"moderate" partition size $J(\eps) \sim 1/\eps$, one can possibly improve \eqref{L2L4bound} to yield the following convergence speeds, valid for $q=2$, $4$ and $s$ bounded, 
\begin{equation}
\label{sqrteps}
\|Y_s^\eps - V_s\|_q \sim \sqrt{\eps} .
\end{equation}
For $q=2$, this is indeed implied by  \eqref{L2L4bound}.
However, 
for $q=4$, our theoretical bound \eqref{L2L4bound} seems to overestimate the size of the partition $J(\eps)$ required to achieve the $L^4$ speed of convergence $\sim\sqrt{\eps}$ 

In this paper, to validate numerically the effective $L^2$ and $L^4$ speeds of convergence of realized volatilities $Y_t^\eps$ to true volatilities $V_t$ we have carried out the following intensive simulations for joint Heston SDEs.

\subsection{Outline of our Heston SDEs simulations} 
\label{simulations} 

We have numerically simulated the joint Heston SDEs with the following specific parameters
\begin{equation}
\kappa = 1.7, ~~~\theta = 4, ~~~\gamma=2, ~~~\mu=0.05, 
\label{params}
\end{equation}
and for 3 values $\rho=0, ~0.3, ~0.7$ of the correlation coefficient between the Brownian noises driving the joint Heston SDEs. The Feller condition is valid since $2 \kappa \theta / \gamma^2 = 3.4$.
To emulate asymptotics as $\eps \to 0$, we consider the partition sizes
\begin{equation}
\label{JJ}
J(\eps)=10, \, 40, \, \eps^{-1}, \, \eps^{-2}
\end{equation}
and five values of $\eps= 0.1$, $0.05$, $0.02$, $0.01$, $0.005$.
Simulations with the fixed partition sizes $J=10, \, 40$ are presented here to illustrate that the errors in the 
approximation of volatility by the realized volatility will not decay as $\eps\to 0$ if the partition has a fixed 
number of points. Numerical simulations with $J=\eps^{-1}, \, \eps^{-2}$ are more interesting since they provide an insight into the convergence rate and selecting the optimal sub-sampling regime for parameter estimation
under indirect  observability. 

Simulations of true volatility paths $V_t$ are implemented by an Euler dicretization scheme for SDEs, with time step $10^{-6}$, except for $\eps=0.005$, $J(\eps)=1/\eps^2$, where the time step was $1.25 \times 10^{-7}$.
We perform a Monte-Carlo simulation by generating $200,000$ independent simulated paths $\{ V_t, Y_t^\eps \}$.
We then partition these $200,000$ paths into sub-ensembles, as discussed in section \ref{numerLqbound}.

\begin{figure}[ht]
\centerline{\includegraphics[width=12cm]{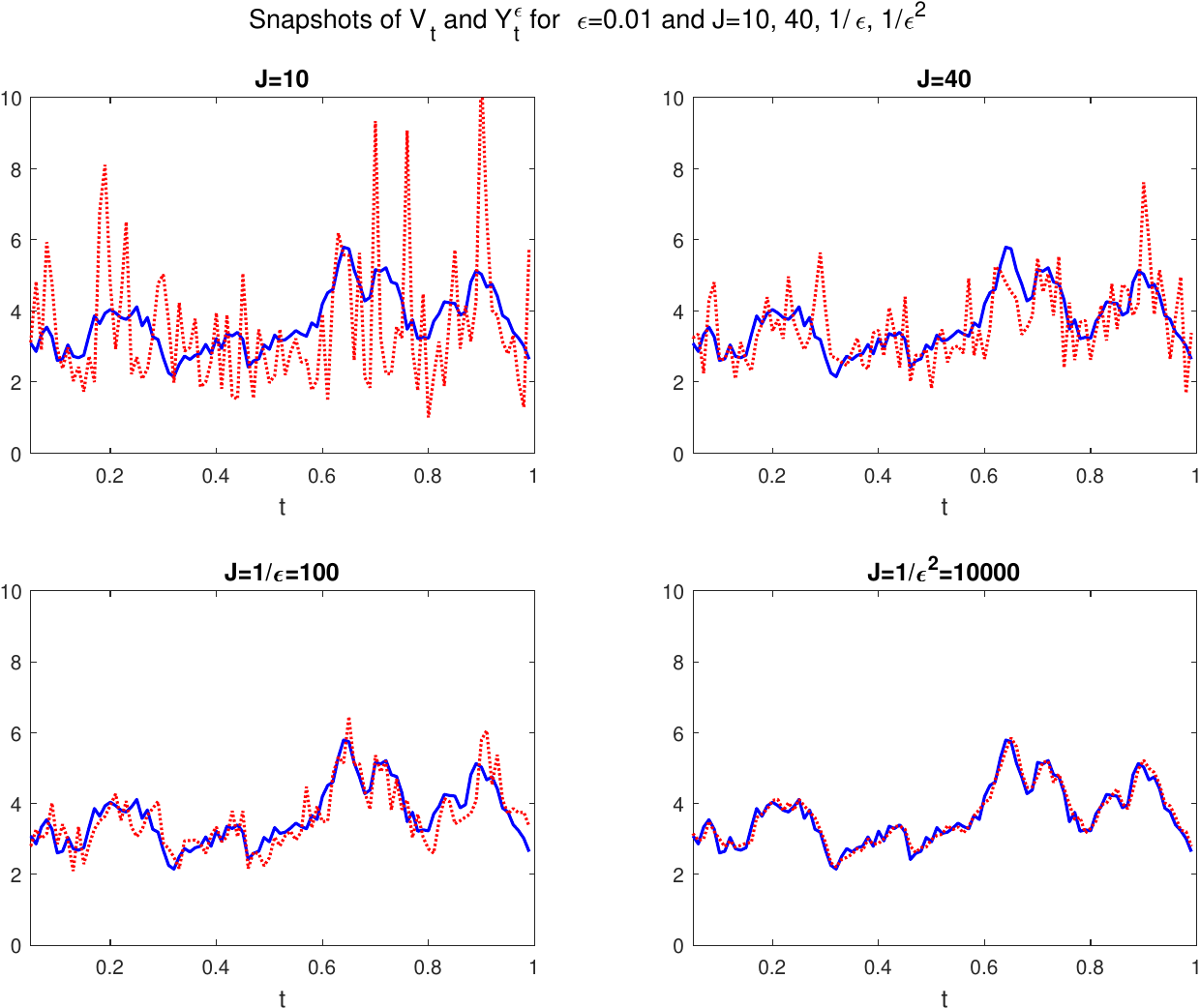}} 
\caption{Volatility $V_t$ and Realized Volatility $Y_t^\eps$ snapshots for $\eps=0.01$ and four partition sizes $J =J(\eps)$ as in \eqref{JJ}. Each sub-plot displays in solid blue the time evolution of one single random trajectory of the volatility $V_t$, $0 \leq t \leq 1$ and displays in dotted red an associated random time evolution of the realized volatility $Y_t^\eps$. Parameters in the Heston volatility SDEs are as in \eqref{params}. The two Heston SDEs are here driven by Brownian motions with zero correlation $\rho=0$.}
\label{fig0}
\end{figure}

\subsection{Snapshots of joint sample paths $\{ V_t, Y_t^\eps \}$}
\label{snapshots} 
For $\eps=0.01$ and $\rho=0$, Figure \ref{fig0} displays four examples of joint paths $\{ V_t, Y_t^\eps \}$, where realized volatilities $Y_s^\eps$ are successively computed with the four $J(\eps)$ listed in \eqref{JJ}.
Clearly, the accuracy of the approximation of $V_t$ by $Y_t^\eps$ increases drastically for larger partition sizes $J(\eps)$. The smallest $J(\eps)$, equal to 10, generates many quite significant inaccuracies for $Y_t^\eps - V_t$. For $J=40$, we still note several significant inaccuracies. But for $J(\eps)=10000$ the sample paths of $V_t$ and $Y_t^\eps$ nearly coincide. Such large partition sizes are generally not feasible: for intraday stock prices sampled every minute, a partition size $J=10000$ would require an unrealistic sliding time window of about 20 trading days; for stock prices sampled every second, such partition would require a sliding window of approximately 2.7 hours.
For small values of $J(\eps)$, a practical remedy to eliminate large sharp peaks of $|Y_t^\eps - V_t|$ is time smoothing of the $Y_t^\eps$ either directly, or by using a weighted average in \eqref{RVformula}.

\begin{figure}[ht]
\centerline{\includegraphics[width=7cm]{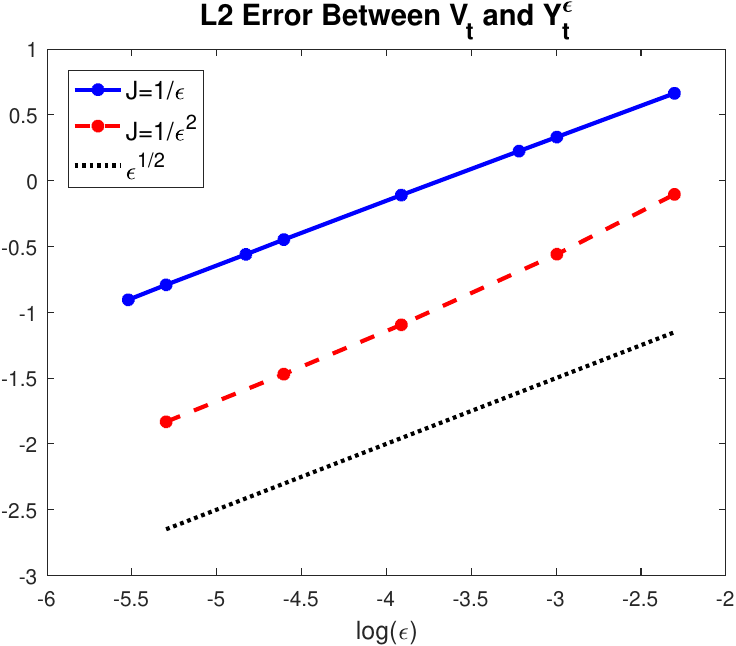}
\includegraphics[width=7cm]{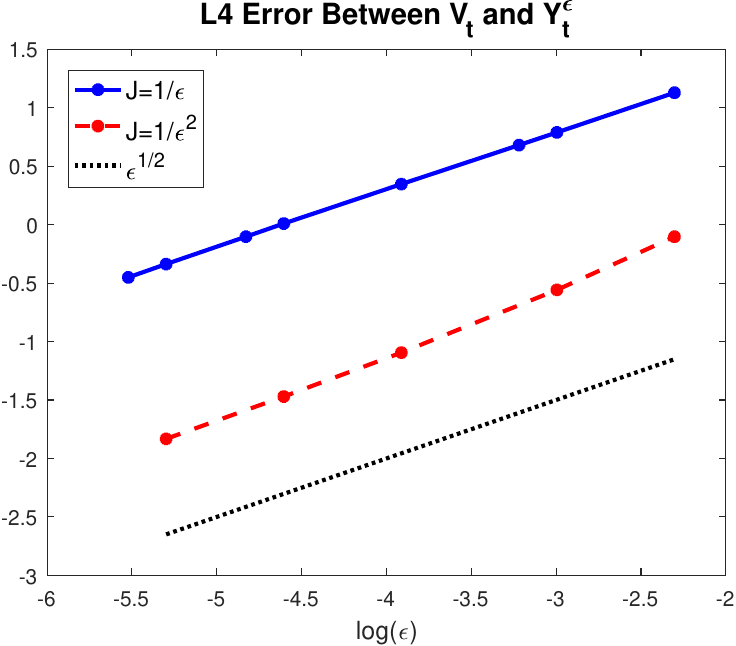}}
\caption{
Simulations of the Heston SDEs with $\rho=0$ and parameters listed in \eqref{params}.
Log-log plots of $L^2$ and $L^4$ errors for $T=1$. We plot $\log(\Lhat^2)$ on the left sub-plot and 
$\log(\Lhat^4)$ on the right sub-plot, as functions of $\log(\eps)$, for the partition sizes $J(\eps)=\eps^{-1}$ (solid blue line) and $J(\eps)=\eps^{-2}$ (dashed red line). The dotted black line represents a reference line with the slope 1/2.}
\label{fig3}
\end{figure}

\begin{table}
\begin{center}
\begin{tabular}{|l|c|c|c|c|c|}
\hline
\hspace*{2.5cm} $\eps = $ & 0.005 & 0.01 & 0.02 & 0.05 & 0.1 \\
\hline
$\Lhat^2$, $J=10$ & $1.85 \pm 0.13$ & $1.85 \pm 0.14$ & $1.86 \pm 0.14$ & $1.88 \pm 0.14$ & $1.94 \pm 0.14$\\
$\Lhat^4$, $J=10$ & $2.98 \pm 0.58$ & $3.00 \pm 0.6$ & $3.01 \pm 0.63$ & $2.99 \pm 0.52$ & $3.09 \pm 0.58$ \\
\hline
$\Lhat^2$, $J=40$ & $0.95 \pm 0.07$ & $0.96 \pm 0.07$ & $0.98 \pm 0.07$ & $1.05 \pm 0.07$ & $1.15 \pm 0.08$\\
$\Lhat^4$, $J=40$ & $1.52 \pm 0.24$ & $1.53 \pm 0.24$ & $1.56 \pm 0.27$ & $1.63 \pm 0.22$ & $1.76 \pm 0.24$ \\
\hline
$\Lhat^2$, $J=1/\eps$ & $0.45 \pm 0.03$ & $0.64 \pm 0.05$ & $0.90 \pm 0.07$ & $1.39 \pm 0.09$ & $1.94 \pm 0.15$\\
$\Lhat^4$, $J=1/\eps$ & $0.71 \pm 0.1$ & $1.01 \pm 0.15$ & $1.41 \pm 0.24$ & $2.20 \pm 0.38$ & $3.09 \pm 0.58$ \\ 
\hline
$\Lhat^2$, $J=1/\eps^2$ & $0.16 \pm 0.008$ & $0.23 \pm 0.013$ & $0.34 \pm 0.019$ & $0.57 \pm 0.033$ & $0.90 \pm 0.056$ \\
$\Lhat^4$, $J=1/\eps^2$ & $0.23 \pm 0.017$ & $0.32 \pm 0.028$ & $0.48 \pm 0.041$ & $0.83 \pm 0.074$ & $1.34 \pm 0.16$ \\
\hline
\end{tabular}
\end{center}
\caption{Values of estimated $\Lhat^2$ and $\Lhat^4$ errors as defined in \eqref{l24num}
in numerical simulations with $J(\eps)=10$, $40$, $1/\eps$, and $1/\eps^2$. 95\% confidence intervals are indicated with ``$\pm$'' numbers.}
\label{table1}
\end{table}
%

%
\subsection{Numerical Asymptotics of $\|Y_t^\eps - V_t\|_q$ as $\eps \to 0$}
\label{numerLqbound}

We partition all $200,000$ simulated paths into sub-ensembles of 
$1000$ paths resulting in $G=200$ of such sub-ensembles. This allows us to compute confidence intervals
for the numerically estimated $L^q$ errors between the volatility and realized volatility.
In particular,  
we fix $T =1$, and for each sub-ensemble 
the empirical mean $M_k(q)$ of $ | Y_T^\eps - V_T |^q$ , $k=1,\ldots,G$ 
provides an estimator $\lhat^q_k = \left( M_k(q) \right)^{1/q}$ for the $L^q$ errors
$\|Y_T^\eps - V_T\|_q$. Final estimates for these $L^q$ errors are then given by 
\begin{equation}
\Lhat^q = \frac{1}{G}\sum_{k=1}^{G} \lhat^q_k, 
\label{l24num} 
\end{equation}
with 95\% confidence intervals $\Lhat^q \pm 1.96 \sigma(q)$ where 
\[
\sigma^2 (q) = \frac{1}{G}\sum_{k=1}^{G} \left( \lhat^q_k - \Lhat^q \right)^2.
\]

Numerical results for the $L^2$ and $L^4$ convergence are presented in Table \ref{table1} and Figure \ref{fig3}.
For constant partition sizes $J(\eps)=10$, $40$ both $\Lhat^2$ and $\Lhat^4$ error estimates
are nearly constant (independent of $\eps$), as predicted by our analytical bound \eqref{L2L4bound}.
Moreover, $\Lhat^2$ as well $\Lhat^4$ errors are both approximately twice smaller for $J=40$ than for $J=10$. For $q=2$, this is correctly predicted by our theoretical bound \eqref{L2L4bound}. 
But for $q= 4$, our theoretical bound is too pessimistic, since it predicts that $\Lhat^4$ should be about $1.4$ times smaller for $J=40$ than for $J=10$.

Figure \ref{fig3} depicts the $L^2$ and $L^4$ errors on the log-log scale for partition sizes $J(\eps)=\eps^{-1}$ and $J(\eps)=\eps^{-2}$. The graphs of estimation errors in these two situations are nearly perfect straight lines with slope 1/2 as soon as $\eps$ is small enough. Figure \ref{fig3} demonstrates quite convincingly that the two types of  estimation errors $\Lhat^2$ and  $\Lhat^4$ do scale like  $\eps^{1/2}$ for  $J(\eps) \sim 1 / \eps$ as well as for $J(\eps)\sim1/\eps^{2}$. 
So our numerical simulations of the joint Heston SDEs support the following conjecture about the asymptotic behaviour (as $\eps\to 0$) of the $L^4$ error
\begin{equation}
\label{L4conj}
\| Y_t^\eps - V_t \|_4 \sim \left( J(\eps) ^{-1/2} + \eps^{1/2} \right). 
\end{equation}

Our simulations indicate that for $q=2$ and $q=4 $ convergence speeds $\| Y_t^\eps - V_t \|_q \sim \sqrt{\eps}$ can be achieved for fixed $t$ when the realized volatility $Y_t^\eps$ is computed with partition sizes $J(\eps) \sim 1 / \eps$. 
This also implies that for the partition size $J(\eps) \sim 1 / \eps$, the lagged covariances of realized volatilities should converge to true lagged covariances at $L^2$-speeds $\sim \sqrt{\eps}$. We would like to point out that these results are obtained for
finite $\eps \ge 0.01$. It is extremely time-consuming to extend these results for
smaller values of $\eps$ and it is possible that the asymptotic behavior of $L^q$ errors
might change for $\eps \ll 0.01$. 
%
Surprisingly, this sub-sampling scheme with $J \sim 1/\eps$ 
gives much fast convergence rate for the $L^4$ norm compared 
to our analytical estimates. As discussed above, it is possible that for extremely small values of $\eps$
numerical simulations would become consistent with our analysis and 
yield the convergence rate of $\eps^{1/4}$. 
Improved convergence speed 
might be related to the ratio of constants in our analytical estimated for the $L^q$ speed of convergence for the moment estimates
However, for such small values of 
$\eps$ numerical computations becomes extremely costly. 
For practical values of $\eps$ considered here we obtain estimated convergence rate of $\eps^{1/2}$.

We also performed numerical simulations with $\rho=0.3$ and $\rho=0.7$ (not displayed here for brevity), where $\rho$ is the correlation between the two Brownian Motions $B_t$ and $Z_t$ driving the joint Heston SDEs. Our numerical results with $\rho > 0$ are almost identical to those for $\rho=0$. 
This is consistent with our proof of Theorem \ref{Heston converge}, which explores the autonomous Heston SDE \eqref{HesEq2} driving the true volatility $V_t$, without ever using the Heston SDE \eqref{HesEq1} for the rate of return process. Another key ingredient of our proof is the study of conditional expectations $\bE(Y | X)$ when $X$ and $Y$ are polynomial functions of a finite number of $V_t$ values. Again this analysis does not use the Heston SDE \eqref{HesEq1}. Constants introduced in Theorem \ref{Heston converge} may depend on $\rho$, but our numerical simulations indicate that this dependence is fairly weak. 

\section{Effective convergence speeds for observable estimators of Heston parameters}
\label{paramSpeed}

In this section we evaluate numerical convergence speeds for our observable estimators 
$\hat{\theta}_\eps$, $\hat{\kappa}_\eps$, $\hat{\gamma}_\eps$ of the Heston volatility SDE parameters. Recall that these estimators are based on estimated covariances of realized volatilities.
This set of simulations is performed as outlined in section \ref{simulations} with the following four values of 
$\eps= 0.1$, $0.05$, $0.02$, $0.01$. Realized volatilities are computed with two
different partition sizes 
\begin{eqnarray}
J(\eps) &=& 1 / \eps \text{~~~for~} \eps=0.1, 0.05, 0.02, 0.01,
\label{subs2} \\
J(\eps) &=& 1 / \eps^2 \text{~~for~} \eps=0.1, 0.05, 0.02.
\label{subs3}
\end{eqnarray}
In order to compute estimators we use the sup-sampling regime 
\begin{equation}
\label{subsa}
N(\eps)=50 \eps^{-3/2}, \;  \quad \Delta(\eps) = \eps^{1/2}
\end{equation}
which is a particular case of our general regime in \eqref{subs}.

Numerical estimates for the $L^2$ errors of estimation for the Heston SDE parameters are computed using a Monte-Carlo approach with 1000 long trajectories consistent with the sub-sampling regime outlined above. Each long trajectory yields one set of estimated parameter 
values computed using \eqref{estim}.

The lag $u^\eps$ is chosen to be approximately 0.6. However, since in our discrete formulas the lag is an integer multiple of $\Delta$, i.e. $u^\eps = r \times \Delta(\eps)$ the lag changes slightly for different values of $\eps$. 
The values of the lag for simulations with different values of $\eps$ are chosen to be
$$
u^\eps=[0.64,~0.66,~0.56,~0.6].
$$

\begin{figure}[ht]
\centerline{\includegraphics[width=7.5cm]{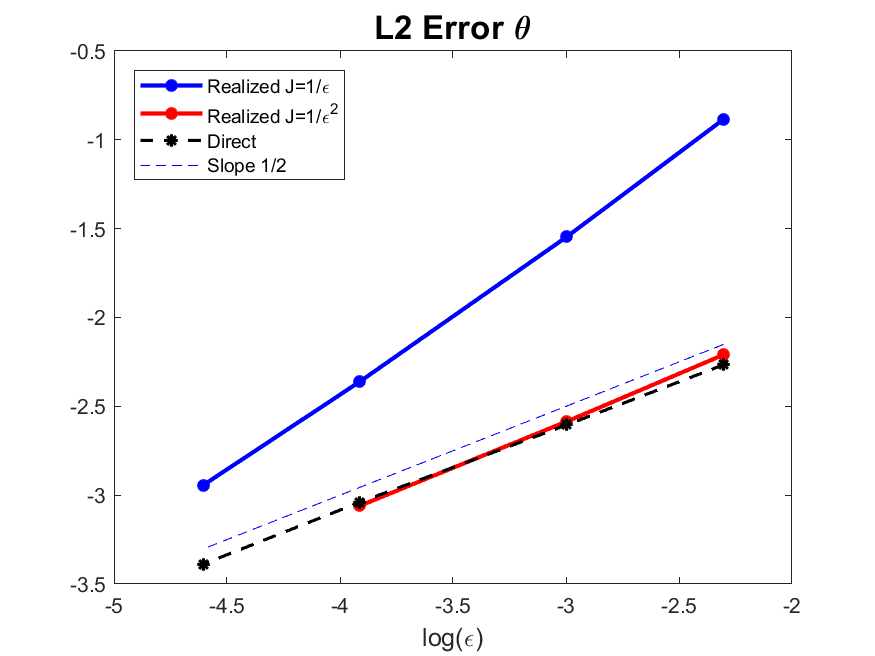}
\includegraphics[width=7.5cm]{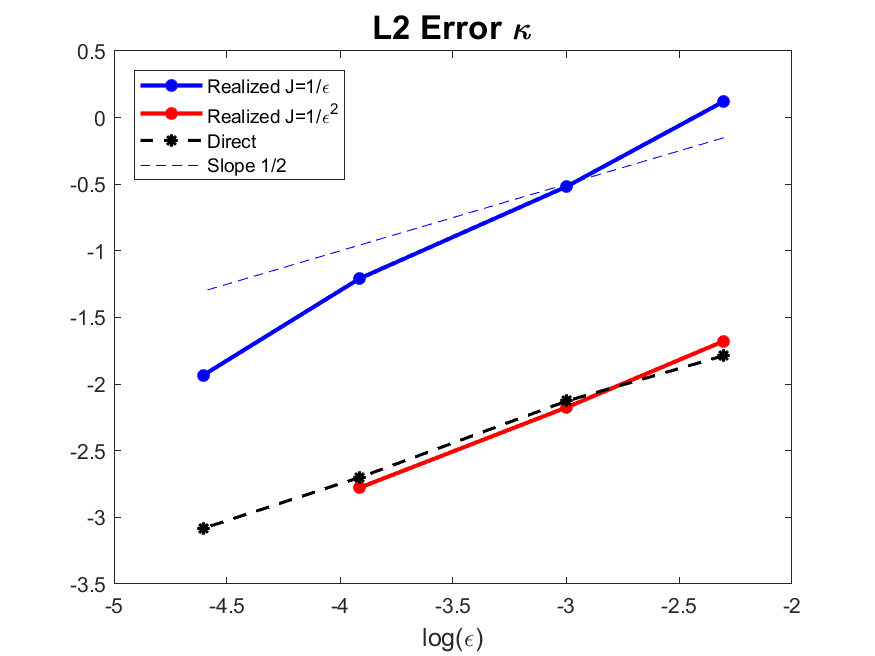}}
\centerline{\includegraphics[width=7.5cm]{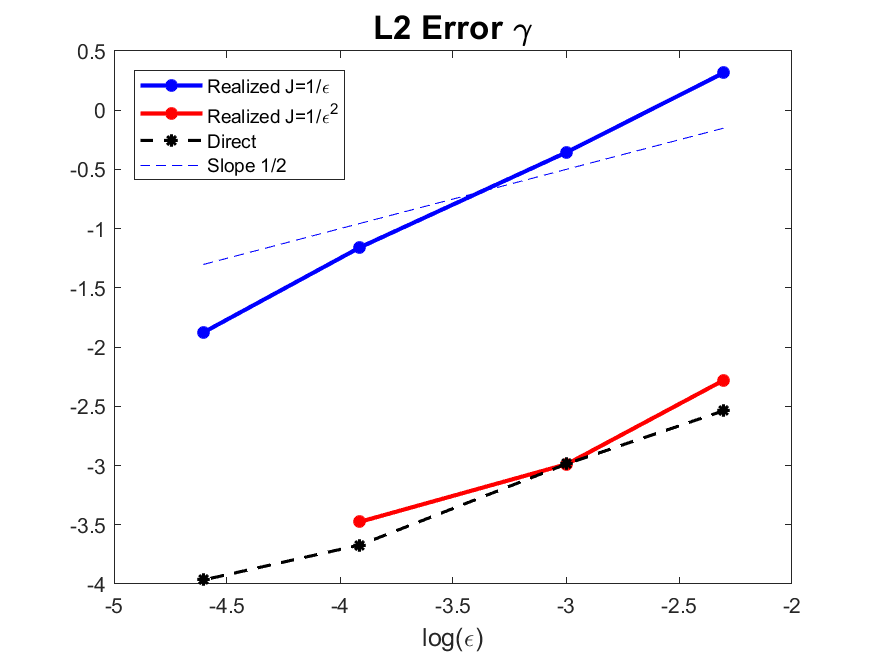}}
\caption{Log-log plots of numerical $L_2$-errors for parameter estimators of the Heston volatility SDE. 
We plot, as functions of $\log(\eps)$, the logarithms of 
$\|\hat{\theta}^\eps - \theta\|_2$ (top left panel), 
$\|\hat{\kappa}^\eps - \kappa\|_2$ (top right panel), 
$\|\hat{\gamma}^\eps -\gamma\|_2$ (bottom panel).
Bold Blue line and Bold Red line - parameter estimators are computed using realized volatility with $J=\eps^{-1}$ and $J=\eps^{-2}$, respectively. 
Black dashed line - parameter estimators computed from direct observations of volatility,
Blue dashed line represents straight reference line with
slope 1/2.}
\label{fig6a}
\end{figure}

Numerical estimates for the $L^2$-errors of parameter estimators are presented in Figure \ref{fig6a}. 
The $L^2$ error $\| \hat{\theta}^\eps -\theta \|_2$ is
depicted in the upper-left part of Figure \ref{fig6a}.
Since $\hat{\theta}^\eps$ estimates the empirical mean of the volatility process,
expression \eqref{bnd0} is directly applicable in this case.
Figure \ref{fig6a} demonstrates that, although the sub-sampling 
regime \eqref{subsa} is identical in both cases, the number of points
for computing the realized volatility, $J$, significantly affects the behavior
of parameter estimators. First, the  numerical error is reduced significantly 
(approximately 10 times)
for $J=\eps^{-2}$ compared to $J=\eps^{-1}$. Second, the 
asymptotic behavior for the $L^2$ error seems also to be affected by
the choice of $J$ which is most evident for parameter  $\theta$.
For the sub-sampling regime \eqref{subsa} and \eqref{subs2}
the decay of $L^2$ error is much faster than $\eps^{1/2}$
for all three parameters. However, with the choice of $J$ in \eqref{subs3}
errors in parameter estimators are almost the same as for the 
estimators computed under direct observability and the error is 
proportional to $\eps^{1/2}$. We would like to point out that numerical 
simulations presented here are for finite values of $\eps\in[0.01,\ldots,0.1]$.
We conjecture that for smaller values of $\eps < 0.01$ the convergence rate
of all parameter estimators computed with $J=\eps^{-1}$ 
should change to $\eps^{-1/2}$ and asymptote to the black line corresponding
to the estimators computed under direct observability.

Our numerical simulations have important practical consequences.
In particular,  our numerical results suggest that it is important to follow the 
regime $J=1/\eps^2$ for larger values of  $\eps$. However, one can switch
to a different regime (e.g. $J \sim \eps^{-3/2}$ or even $J\sim\eps^{-1}$) 
for smaller values of $\eps$ to reduce the computational
overhead. This is motivated by rather fast rate of convergence 
for parameter estimators computed with $J=\eps^{-1}$.

In the regime with $J=\eps^{-1}$ for smaller values of $\eps=0.01, 0.02$ 
errors in all parameter estimators decay significantly compared to $\eps=0.05,0.1$. Behavior of parameter estimators themselves is depicted in 
Figure \ref{fig6b}. It is obvious that the sub-sampling regime $J=\eps^{-1}$
results in very large errors for larger values of $\eps=0.05$, $0.1$.
Parameter $\theta$ is estimated more accurately 
under the computational scheme with $J=\eps^{-1}$ for $\eps=0.01$, but
there is still approximately 10\% relative error in estimating parameters 
$\kappa$ and $\gamma$ in this regime. On the other hand, relative errors in
estimating all three parameters are much smaller for $J=\eps^{-2}$ and 
$\eps=0.1$. Therefore, the most beneficial strategy is to use a bigger 
window, $\eps$, for computing the realized volatility with a large number of
points $J=\eps^{-2}$ for the return process.
\begin{figure}[ht]
\centerline{\includegraphics[width=7.5cm]{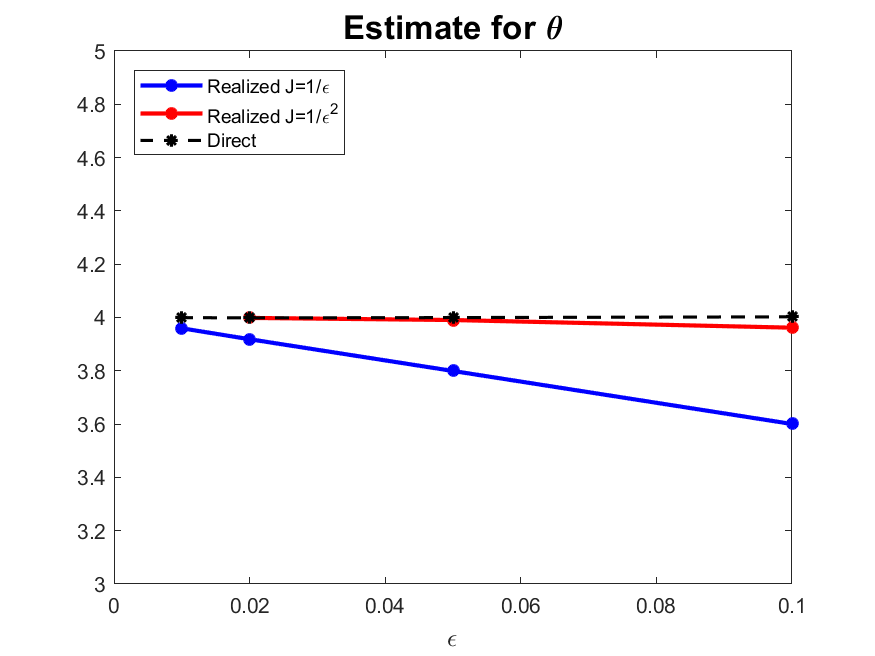}
\includegraphics[width=7.5cm]{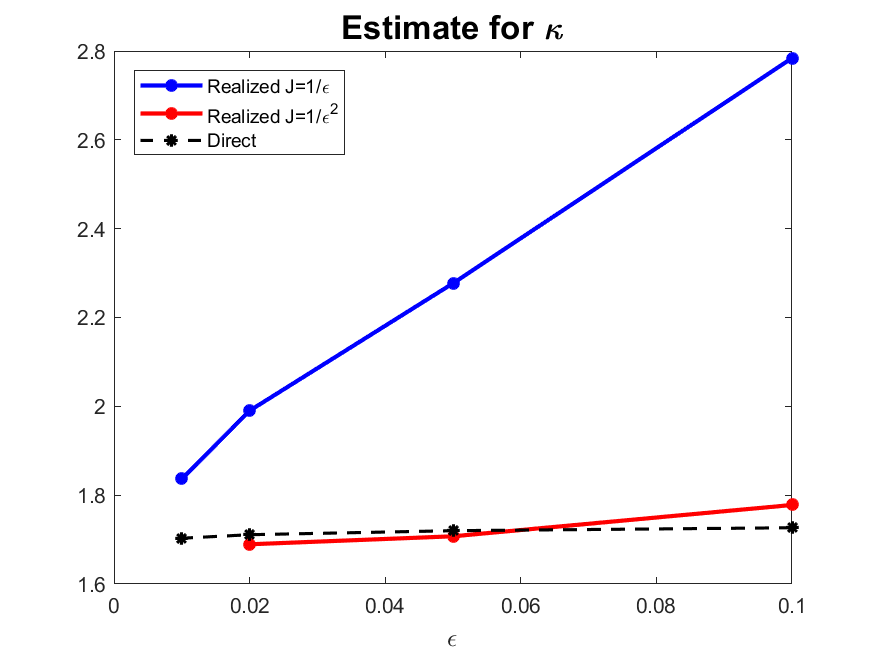}}
\centerline{\includegraphics[width=7.5cm]{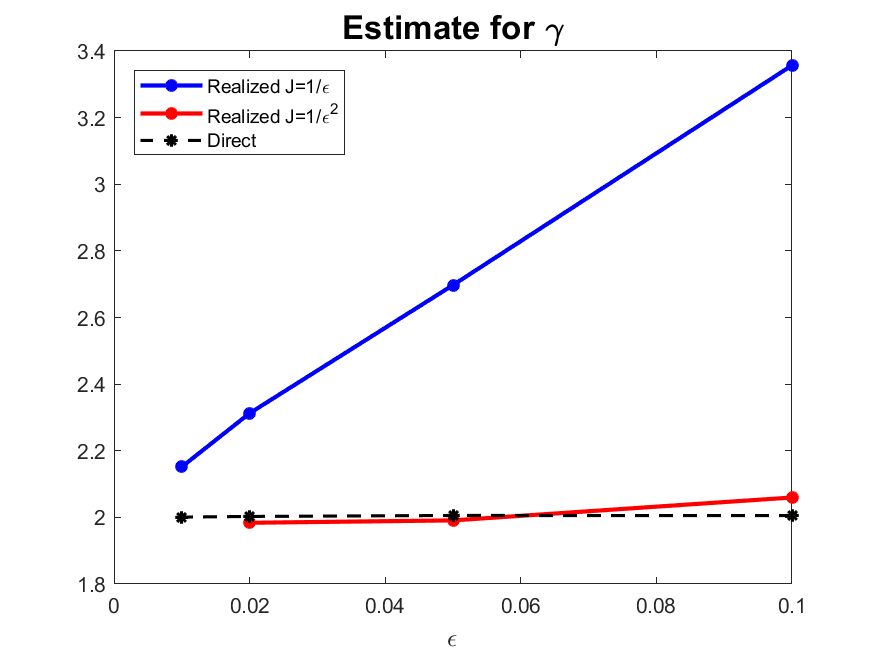}}
\caption{Behavior of parameter estimators of the Heston volatility SDE vs $\eps$. 
Bold Blue line and Bold Red line - parameter estimators are computed using realized volatility with $J=\eps^{-1}$ and $J=\eps^{-2}$, respectively. 
Black dashed line - parameter estimators computed from direct observations of volatility.}
\label{fig6b}
\end{figure}

Asymptotic behaviour of our observable estimators for the Heston parameters $\kappa$ and $\gamma$ strongly depends on the behavior of the lagged covariances $K^\eps(u)$ of $Y_t^\eps$. Thus, we also present our numerical results for the estimation
of $\hat{K}^\eps(u)$ with 
two particular time lags $u=0$ and $u \approx 0.6$.

The mean and lagged covariances of $Y_t^\eps$ are approximated by their empirical estimators, given by 
\begin{equation}
\label{momyt}
\hat{m}^\eps = \frac{1}{N}\sum\limits_{j=1}^N Y_{j\Delta}^\eps, \qquad
\hat{K}^\eps(u) = 
\frac{1}{N-s}\sum\limits_{i=1}^{N-s} Y_{j\Delta}^\eps Y_{(j+s)\Delta}^\eps - (\hat{m}^\eps)^2,
\end{equation}
where for lag $u=0$ the integer $s$ is $s=0$, and is chosen such that $s\Delta \approx 0.6$ when lag $u\approx 0.6$.
Recall that the stationary moments of true volatilities  are given by 
\eqref{m1m2}, \eqref{covstationary}.

The $L^2$ errors for the lagged covariances are computed from Monte-Carlo simulations as
\begin{equation}
\label{l2errY3}
\| \hat{K}^\eps(u) - K(u) \|_2 \equiv \sqrt{\frac{1}{MC}\sum\limits_{k=1}^{MC} \left(\hat{K}^\eps(u) - K(u) \right)^2}, 
\end{equation}
where the sum involves $MC=1000$ independent evaluations of $\hat{K}^\eps(u) $. 
Results for the covariance estimation are displayed in Figure \ref{fig5}. 
\begin{figure}[ht]
\centerline{\includegraphics[width=7.5cm]{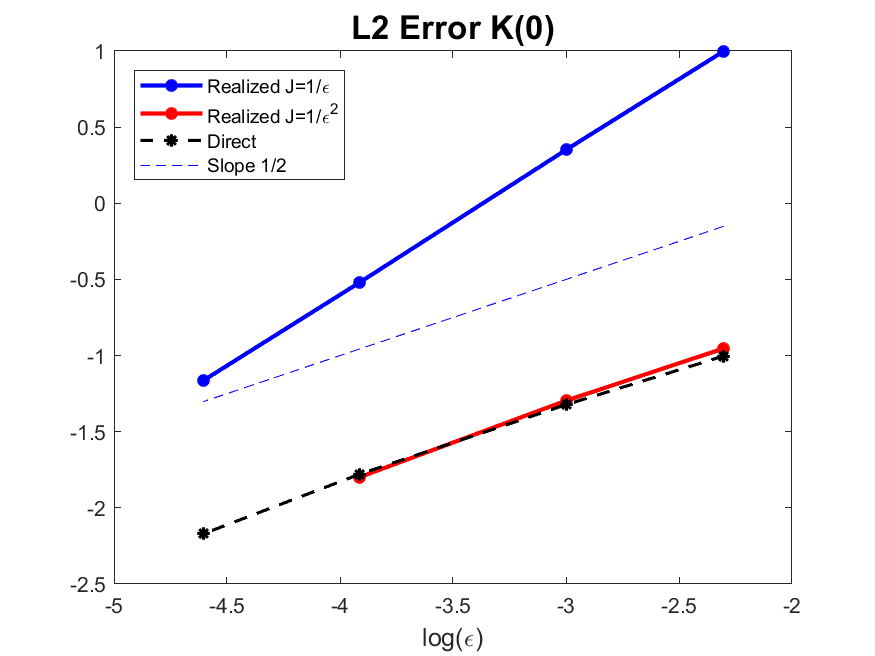}
\includegraphics[width=7.5cm]{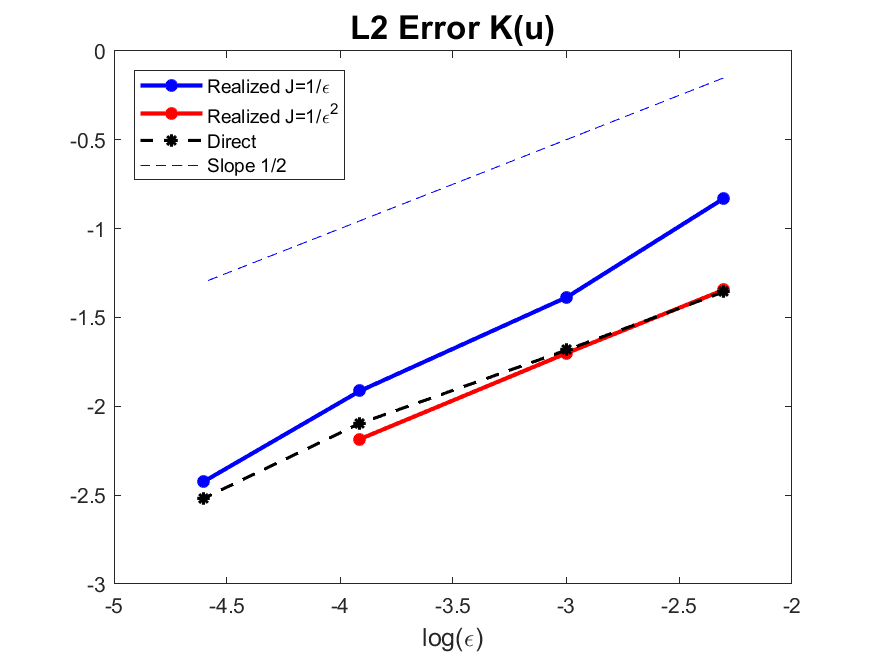}}
\caption{Log-log plots for $L^2$ estimation errors for second moments of $Y_t^\eps$ computed by equations \eqref{l2errY3}.
Left panel - $\log \|\hat{K}^\eps(0) - K(0)\|_2$, 
Right panel - $\log \|\hat{K}^\eps(u) - K(u)\|_2$ with $u \approx 0.6$.}
\label{fig5}
\end{figure}
Behavior of $L^2$ errors for estimated second moments is consistent with the
behavior of parameter estimators discussed earlier. In particular, for the range of 
$\eps\in[0.01,\ldots,0.1]$ convergence rate of $\hat{K}^\eps(u)$
computed
with $J=\eps^{-1}$ appears to be much
faster than $\eps^{-1/2}$, especially for $\hat{K}(0)$. 
Similar to the behavior of parameter estimators, 
we conjecture that this is due to the finite range of $\eps$.

The choice of the lag $u^\eps$ is motivated by some practical considerations. In particular, one should perform an a-posteriori check after the parameter estimator $\hat\kappa$ is computed and ensure that the estimated lagged correlation $K(u^\eps)$ is not too close to 0 or 1, for instance by checking that $e^{-\hat{\kappa} u^\eps}$ lies between $0.3$ and $0.7$.
Apart from such practical constraint above, the choice of $u^\eps$ is otherwise arbitrary. 
We performed numerical simulations (not reported here) investigating several other choices of the lag $u^\eps$. In particular, we considered $u^\eps \approx 0.3$ and the ``vanishing lag'' case $u^\eps = \Delta \equiv \sqrt{\eps}$. Our numerical simulations indicate that for the specific Heston SDE parameters considered here the choice $u^\eps \approx 0.6$ yielded  near-optimal asymptotic behavior of both, observable moments estimators and parameter estimators.

%
\section{Conclusions}
\label{conc}

We carried out an extensive analytical and numerical investigation of the Heston joint SDEs driving jointly the squared volatilities $V_t$ and the rate of returns $R_t$. Since the volatility process $V_t$ cannot be observed directly, realized volatilities $Y_t^\eps$ computed from the return process $R_s$ with $s$ in the sliding window $\left[t-\eps, t\right]$ provide classical observable approximations of the unobservable $V_t$.

The main goal of this paper is to define and study observable estimators of the Heston SDEs parameters computed from the $Y_t^\eps$, and exhibiting asymptotic consistency as $\eps \to 0$. 
This context fits our general framework of \emph{indirect observability} where parameter estimators for the dynamics of an unobservable process $X_t$ can only be computed from observations of a process $Y_t^{\eps}$ approximating $X_t$ as $\eps \to 0$. 
Computing realized volatilities $Y_t^\eps$ from the rates of returns $R_t$ requires partitioning the window $\left[t-\eps,t\right]$ into $J(\eps)$ time intervals. For the Heston SDEs we prove precise bounds for $L^q$ norms $\|Y_t^\eps - V_t\|_q $ in terms of $J(\eps)$ and $\eps$. In particular we show that $\|Y_t^\eps - V_t\|_4 \leq C \sqrt{\eps}$ provided $J(\eps) \sim 1/\eps^2$. However, 
for small window sizes, $\eps$,
partition sizes  $J(\eps) \sim 1/\eps^2$ are not very practical  
since they require an overwhelming number of points for small window size $\eps$. 
Our numerical simulations indicate that it is possible to obtain 
reasonable numerical estimates in $L^2$ sense with
more practical partition sizes $J(\eps) \sim 1/\eps$. However, 
$L^4$ errors $\|Y_t^\eps - V_t\|_4$ are more sensitive to the choice of the
partition size.

Our observable estimators of the Heston SDEs parameters are defined as  explicit functions of the empirical mean and two empirical lagged covariances computed from $N(\eps)$ observations $Y_{j \Delta(\eps)}^\eps$, $j=1,\ldots,N(\eps)$ of the realized volatility, sub-sampled with time step $\Delta(\eps)$. 
We prove that for fastest convergence speed of observable parameter estimators to true parameters, the optimal sub-sampling regime is provided by $N(\eps) \sim \eps^{-3/2}$ and $\Delta(\eps) \sim \eps^{1/2}$ with $J(\eps) \sim \eps^{-2}$. 
Our sub-sampling scheme \eqref{optimal1} provides a needed balance between the $L^2$ errors of estimation  on empirical covariances and the $L^2$ difference between true and realized volatilities. 
This optimal sub-sampling scheme 
corresponds to a total  observational time  $T(\eps) = N(\eps)  \Delta(\eps)  \sim 1 / \eps$ and a  total
number $n(\eps)$ of observed returns rate values $n(\eps) = N(\eps) J(\eps) = 1 / \eps^{7/2}$.

Surprisingly, our numerical simulations indicate a much faster speed
of convergence for the sub-sampling regime  $N(\eps) \sim \eps^{-3/2}$ and $\Delta(\eps) \sim \eps^{1/2}$ with $J(\eps) \sim \eps^{-1}$ for a rather wide
range of $\eps\in[0.01,0.1]$. Convergence rates for all three parameter estimators
are close to $\eps^1$. Parameter estimators computed under indirect observability
are inferior to estimators computed from directly observed time-series 
of the volatility process. Therefore, we conjecture that the convergence rate
for the sub-sampling regime with $J=1/\eps$ should change for $\eps \ll 0.01$.
However, verifying this with numerical simulations is extremely computationally 
costly. In addition, we also observe that relative errors under the sub-sampling
regime $J=1/\eps$ are much larger compared to the sub-sampling regime
with $J=1/\eps^2$. 
Therefore, to reduce the computational cost, the
\textbf{optimal estimation strategy}
is to use a larger window $\eps$ for computing the realized volatility
with a large number of points $J=1/\eps^2$ for the return process.

When one imposes a bound on the total observational time $T(\eps)$, our theory and numerical simulations indicate that there is a lower $L^2$ bound on the estimation errors for the  parameters of the Heston  volatility SDE. An upper bound on $T(\eps)$ essentially forces a lower-bound on $\eps$. Therefore, in practice, it is then not beneficial to keep over-refining the partition  of the sliding time window used to  compute realized volatilities.
Moreover, when $T(\eps)$ is bounded, decreasing the size of the sliding window, $\eps$, constrains the number of observations of the return process inside this window to decrease, and this generates more inaccurate approximations of true volatilities by realized volatilities.

Our theoretical analysis and numerical simulations of the Heston SDEs presented here provide practical guidelines for fitting joint Heston SDEs to practical observations of stock prices. In particular our results should help define  adequate choices for  the size $\eps$ of the sliding windows used  to compute realized volatilities, as well as  for the selection of an efficient sub-sampling time step  of returns rate observations. 
\smallskip

{\bf Acknowledgements.}
I.T. and R.A. were supported in part by the NSF Grant DMS-1109582.
I.T. is also partially supported by the NSF Grant DMS-1620278.


\appendix

%

\section{Polynomial functions of volatilities and Theorem \ref{polymoments}}
\label{polyfunctionals}
We evaluate conditional moments for polynomial functions of squared volatilities $V_t$. Let $\calF_s $ be the filtration generated by the Brownian $B_t$ driving the  Heston volatility SDE. Note that conditioning by $\calF_s$ gives the same results when the volatility process starts at any fixed $V_0 =y >0$  or when it is the only stationary  process driven by the volatility Heston SDE. 

Recall the statement of Theorem \ref{polymoments}.
Fix any polynomial $h$ of total degree $n$ in $k$ variables $(x_1, \ldots, x_k)$. 
Let $0 = u(0) < u(1) < \ldots < u(k)$ be any sequence of $k+1$ lag instants.
For $T >0$, define random variables $H$ and $H_T$ by
\begin{equation}\label{defH} 
H= h \left(V_{u(1)}, \ldots, V_{u(k)} \right) \;\; \text{and} \; 
H_T= h \left(V_{u(1)+T}, \ldots, V_{u(k)+T} \right).
\end{equation} 
Recall that $\nu_T = e^{- T \kappa}$. 
Define $w_j = e^{ - \kappa (u(j+1) - u(j))}$ for $j = 0, \ldots, k-1$. 
There is then a polynomial $POL$ in $k+2$ variables such that for all $T>0$ and all $y>0$
\begin{equation}\label{EHTa}
\bE_y (H_T ) = POL(\nu_T, y \nu_T, w_0, w_1, \ldots , w_{k-1}) .
\end{equation}
The degree and coefficients of POL are determined by the integers $n,k$, the coefficients of $h$, and the vector $\btheta$. 
The asymptotic polynomial moments are then given by
$$
\lim_{T \to \infty} \bE_y (H_T) = \bE_\psi (H) = POL(0, 0, w_0, w_1, \ldots , w_{k-1}).
$$
For any integer $q \geq 1$ there is a positive constante $C$, and an integer $p \geq 1$, determined only by $q, k,\btheta$ and the polynomial $h$ such that, for all positive $T$ and $y$, and all $0 = u(0) < u(1) < \ldots < u(k)$
\begin{equation}\label{Lqlimitpola}
\Big| \bE_y \left[ | H_T - \bE_\psi (H) |^q \right] \Big| \leq C (1 + y^p) e^{- T \kappa}.
\end{equation}
In particular for $q= 1$ one has
\begin{equation}\label{limitpola}
\Big| \bE_y (H_T) - \bE_\psi (H) \Big| \leq C (1 + y^p) e^{- T \kappa}.
\end{equation}

\noindent
\textbf{Remarks.} Equation \eqref{Lqlimitpola} also implies that as $T \to \infty$, the random polynomial functionals $H_T$ converge in $L^q$-norm to the constants $\bE_\psi(H)$, where $L^q$-norms are computed under $\bE_y$. 
Note also, that all the constants introduced in the theorem and in its proof below do not depend on the actual lags $u(0) < u(1) < \ldots < u(k)$. 

\paragraph{Proof of Theorem \ref{polymoments}:} 
\begin{proof}
By linearity, we only need to consider the case when $h$ is a monomial in $k$ variables. For $k=1$, the result was proved by \eqref{Mpol}. Proceeding by recurrence on $k$, assume the result is true for monomials in $k-1$ variables $(x_2, \ldots, x_k)$. Any monomial $h$ in $k$ variables can be written as
$h = x_1^m g(x_2, \ldots, x_k)$.
Define 
$$
G_T= g \left(V_{u(2)+T}, \ldots, V_{u(k)+T} \right) \;\; \text{and} \;\; H_T= V_{u(1)+T}^m G_T.
$$
The recurrence hypothesis provides a polynomial $R$ in $(k+1)$ variables such that, for all $T$
$$
\bE_y (G_T) = R \left(\nu_T, y \nu_T, w_1, w_2, \ldots , w_{k-1} \right) 
$$
where the coefficients of $R$ are determined by $g, \btheta$. By the Markov property we thus get 
$$
\bE \left(G_T | \calF_{u(1) +T} \right) = R \left(\nu_T, V_{u(1) +T} \nu_T, w_1, w_2, \ldots , w_{k-1} \right) .
$$
Since $\bE_y [H_T] = \bE_y [ V_{u(1)+T}^m \bE(G_T | \calF_{u(1) +T}) ] $ we then obtain
$$
\bE_y [H_T] = \bE_y \left[ V_{u(1)+T}^m R(\nu_T, V_{u(1) +T} \nu_T, w_1, w_2, \ldots , w_{k-1}) \right].
$$
Each monomial $M$ of $R$ is of the form $\nu_T^p (V_{u(1) +T} \nu_T)^j S(w_1, w_2, \ldots , w_{k-1})$ for some $p$, $j$ and some polynomial $S$. Then in the right-hand side of \eqref{EHTa}, $M$ contributes a term of the form
$$
\Gamma(M) = \nu_T^{p+j} S(w_1, w_2, \ldots , w_{k-1}) \bE_y \left[ V_{u(1)+T}^{m+j} \right].
$$ 
Due to \eqref{Mpol} with $q= m+j$, this last conditional expectation is a polynomial in the two variables 
$$
\nu_{u(1)+T} = \nu_T w_0 \;\; \text{and} \; V_0 \nu_{u(1)+T} = y \nu_T w_0
$$
with coefficients depending only on $m+j$ and $\btheta$. 
Hence $\Gamma(M)$ is a polynomial in $\nu_T$ and $y \nu_T$, with coefficients which are polynomials in $(w_0, w_1, w_2, \ldots , w_{k-1}) $, fully determined by $m$, $j$, $\btheta$. The same property must then hold for the sum $\bE_y(H_T)$ 
of all the $\Gamma(M)$ contributed by the monomials $M$ of $R$.
This completes the proof of \eqref{EHTa} by recurrence on $k$.

Write $POL$ in \eqref{EHTa} as a polynomial $POL(z)$ in the $k+2$ variables $z_i$.
The vector $z(T) = (\nu_T, y \nu_T, w_0, \ldots, w_{k-1})$ tends to $z(\infty) = (0, 0, w_0, \ldots, w_{k-1})$ as $T \to \infty$. 
The polynomial $Q(z(T)) = POL(z(T)) - POL(z(\infty))$ can be written for some integer $p$ 
$$
Q (z(T)) = \nu_T A_0 + \sum_{s=1}^p y^s \nu_T^s A_s,
$$
where for $s=0, \ldots, p$, each $A_s$ is a polynomial in the $(k+1)$ variables $(\nu_T, w_0, \ldots, w_{k-1})$. Since all these positive $(k+1)$ variables are inferior to $1$, then each $|A_s|$ remains bounded for all $T \geq 0$ and all $u(0) < u(1) < \ldots < u(k)$. Hence there is a constant $C$ such that
$$
| A_s | \leq C \;\; \text{and} \;\; y^s \leq C (1+y^p) \;\; \text{for all} \;\; s= 0, \ldots, p, \, T \geq 0, \, y > 0.
$$
For all $s \geq 1$ we have $\nu_T^s \leq \nu_T= e^{-T \kappa}$, and hence the expansion of $Q(z(T))$ provides a new constant $C_1$ such that, for all $u(0) < u(1) < \ldots < u(k)$, 
$$
\big| \bE_y (H_T) - \bE_\psi (H) \big| = \big| Q (z(T)) \big| \leq C_1 (1+y^p) e^{-T \kappa} \;\; \text{for all} \;\; T \geq 0, \, y > 0.
$$
This proves \eqref{limitpola}.

Let $\barH = \bE_\psi (H)$. 
Expand $\beta(T) = (H_T - \barH)^q$ as a linear combination of terms of the form 
$ \barH^{q-j} H_T^j $ for $j= 0, \ldots, q$.
Recall that $h$ is a polynomial in $x_1, x_2, \ldots, x_k$. For $j$ fixed, $\sigma_j = h^j$ is also a polynomial in $x_1, x_2, \ldots, x_k$. By definition \eqref{defH}, we can express both $\Sigma = H^j$ and $\Sigma_T = H_T^j$ as 
$$
\Sigma= \sigma_j \left(V_{u(1)}, \ldots, V_{u(k)} \right) \;\; \text{and} \;\; 
\Sigma_T= \sigma_j \left(V_{u(1)+T}, \ldots, V_{u(k)+T} \right).
$$ 
For each $j$, equation \eqref{limitpola} applied to the polynomial $\sigma = h^j$ provides a constant $C_j$ and an integer $p(j)$ such that 
$$
\Big| \bE_y (\Sigma_T) - \bE_\psi (\Sigma) \Big| \leq C_j \left(1+y^{p(j)} \right) e^{-T \kappa} \;\; \text{for all} \;\; T \geq 0, \, y > 0
$$
and hence there are constants $c_j$ such that 
$$
\Big| \bE_y (\bar{H}^{q-j}H_T) - \bE_\psi (\bar{H}^{q-j} H_T)) \Big| \leq c_j \left(1+y^{p(j)} \right) e^{- T \kappa} \;\; \text{for all} \;\; T \geq 0, \, y > 0.
$$
Applying this to $j= 0, \ldots, q$ and using the Newton binomial formula yields, for some new constant $C$,
$$ 
\bE \left[ \big| \beta(T) \big| \right] \leq C e^{- T \kappa} \sum_{j= 0}^q c_j \left(1+y^{p(j)} \right) 
\;\; \text{for all} \;\; T \geq 0, \, y > 0
$$
which completes the proof of \eqref{Lqlimitpola}.

\qed
\end{proof}


\end{document}